\documentclass[a4paper,final]{article}
\usepackage[utf8]{inputenc}
\usepackage{amsmath,amsthm}
\usepackage[charter]{mathdesign}
\usepackage{bbm}
\usepackage{braket}
\usepackage{enumerate}
\usepackage{hyperref}

\newtheorem{theorem}{Theorem}
\newtheorem{lemma}[theorem]{Lemma}
\newtheorem{proposition}[theorem]{Proposition}
\newtheorem{corollary}[theorem]{Corollary}

\theoremstyle{definition}
\newtheorem{definition}[theorem]{Definition}
\newtheorem{example}[theorem]{Example}
\newtheorem*{conjecture}{Conjecture}
\newtheorem*{dfn}{Definition}

\theoremstyle{remark}
\newtheorem*{remark}{Remark}

\newcommand{\bmaps}{\ensuremath{\mathrm{B}}}
\newcommand{\complexes}{\ensuremath{\mathbbm{C}}}
\newcommand{\reals}{\ensuremath{\mathbbm{R}}}
\newcommand{\malg}[1]{\ensuremath{M_{#1}(\complexes)}}
\newcommand{\bbone}{\ensuremath{\mathbbm{1}}}
\newcommand{\id}{\ensuremath{\mathrm{id}}}
\newcommand{\transpose}{\ensuremath{\tau}}
\newcommand{\trpse}{\ensuremath{^{\scriptscriptstyle\mathrm{T}}}}
\newcommand{\ptranspose}{\ensuremath{\tau_{\scriptscriptstyle\mathrm{P}}}}
\newcommand{\norm}[1]{\ensuremath{\Vert#1\Vert}}
\newcommand{\abs}[1]{\ensuremath{\vert#1\vert}}
\renewcommand{\set}[1]{\ensuremath{\{#1\}}}
\renewcommand{\Set}[1]{\ensuremath{\bigg\{#1\bigg\}}}
\newcommand{\setdef}{\ensuremath{\,\,\vert\,\,}}
\newcommand{\Setdef}{\ensuremath{\,\,\big\vert\,\,}}
\newcommand{\Cstar}{\ensuremath{C^*}}
\newcommand{\conj}{\ensuremath{^*}}
\newcommand{\cconj}[1]{\ensuremath{\overline{#1}}}
\newcommand{\defeq}{\ensuremath{\operatorname{:\!=}}}
\newcommand{\maps}{\ensuremath{\mathcal{L}}}
\newcommand{\pmaps}{\ensuremath{\maps^+}}
\newcommand{\cpmaps}{\ensuremath{\maps^{\scriptscriptstyle\mathrm{CP}}}}
\newcommand{\diagonal}[1]{\ensuremath{\operatorname{diag}_{#1}(\complexes)}}
\let\diag=\diagonal
\newcommand{\restr}[1]{\ensuremath{\vert_{#1}}}
\newcommand{\hilbert}[1]{\ensuremath{\mathfrak{#1}}}
\newcommand{\algebra}[1]{\ensuremath{\mathfrak{#1}}}
\newcommand{\gebp}{\ensuremath{\ge_{\mathrm{bp}}}}
\newcommand{\symmetries}{\ensuremath{\mathcal S}}

\DeclareMathOperator{\Tr}{Tr}
\DeclareMathOperator{\ad}{Ad}
\DeclareMathOperator{\lspan}{span}

\def\scaled{\let\onleft=\left\let\onright=\right}
\def\unscale{\let\onleft=\relax\let\onright=\relax}

\newcommand{\setdzero}{\ensuremath{\mathfrak D_0}}
\newcommand{\setd}{\ensuremath{\mathfrak D}}

\usepackage{color}

\begin{document}

\title{On the structure of positive maps II: low dimensional matrix algebras}
\author{Władysław A. Majewski\thanks{\small\href{mailto:fizwam@ug.edu.pl}{\nolinkurl{fizwam@ug.edu.pl}}} \qquad
Tomasz I. Tylec\thanks{\small\href{mailto:ttylec@gmail.com}{\nolinkurl{ttylec@gmail.com}}}
\\
{\small Institute of Theoretical Physics and Astrophysics, University of Gdańsk}}
\maketitle

\begin{abstract}
 We use a new idea that emerged in the examination of exposed positive maps
 between matrix algebras to investigate in more detail the difference between
 positive maps on $\malg 2$ and $\malg 3$. Our main tool stems from classical
 Grothendieck theorem on tensor product of Banach spaces and is an older and
 more general version of Choi-Jamiołkowski isomorphism between positive maps
 and block positive Choi matrices. It takes into account the correct topology
 on the latter set that is induced by the uniform topology on positive maps.
 In this setting we show that in $\malg 2$ case a large class of nice positive
 maps can be generated from the small set of maps represented by self-adjoint
 unitaries, $2 P_x$ with $x$ maximally entangled vector and $p\otimes\bbone$
 with $p$ rank 1 projector. We show why this construction fails in $\malg 3$
 case. There are also similarities. In both $\malg 2$ and $\malg 3$ cases any
 unital positive map represented by self-adjoint unitary is unitarily
 equivalent to the transposition map. Consequently we obtain a large family of
 exposed maps. We also investigate a convex structure of the Choi map, the
 first example of non-decomposable map. As a result the nature of the Choi map
 will be explained. This gives an information on the origin of appearance of
 non-decomposable maps on $\malg 3$.
\end{abstract}

\section{Introduction}
\label{sec:introduction}

Positive maps between $n\times n$ matrix algebras play an important role in the
entanglement theory as they can be used do classify entangled states on two
$n$-level quantum systems. Furthermore, it seems that positive (not only
completely positive) maps play an important role in description of some
special dynamical systems (see e.g. \cite{majewski2007non} and
\cite{shaji2005s}). The problem of characterizing all positive maps was
unsolved for over 50 years even for matrix of dimension $n=3$ or higher. Very
recently, a general characterization of unital positive maps, for finite
dimensional case, was given in \cite{Majewski:2010fk}. To complete an
analysis of the structure of positive maps, it is natural to ask a question
what is an essential difference between simple case of maps between $n=2$
matrix algebras and maps between $n=3$ matrix algebras. In that way one hopes
to fully understand the origin of appearance of non-decomposable maps for the
$n=3$ case.

In this paper we will shed some light on this problem using new idea. It
originated from the attempt to characterize exposed points of the set of
positive maps between matrix algebras~\cite{Majewski:2010fk}. In that work the
following set of Choi matrices naturally emerges
\begin{equation*}
  \tilde{\mathfrak D} = \set{ \text{symmetries}, n P_x, p\otimes\bbone},
\end{equation*}
where symmetries are selfadjoint unitaries, $x$ is fully entangled vector on
$\complexes^n\otimes\complexes^n$ and $p$ is some rank one projector. We will
show that this set is rich enough to describe all regular extreme positive
maps in the $n=2$ case (i.e.\ maps with the property that their restriction to
diagonal subalgebra is still extreme, cf. Def.~\ref{dfn:regpos}). We will also
show at which point it fails in the $n=3$ case. To further examine $n=3$ case
we explore relation of Choi matrices given by symmetries to Choi matrix of
transposition map. Finally analysis of the convex structure of Choi map will
be presented. This gives now information on the nature of the first example of
non-decomposable map. 

The article is organized as follows. It the Section~\ref{sec:preliminaries} we
recall some basic notions and introduce useful tools. In the
Section~\ref{sec:extr-points-pmapsd} we consider the case of extremal positive
maps from abelian algebra to $\malg n$. This is our main tool in the
Section~\ref{sec:extr-posit-maps}, where we discuss the role of the set
$\tilde{\mathfrak D}$ in the set of regular extreme maps. In the
Section~\ref{sec:choi_map} we briefly discuss the convex structure of Choi
map. Finally, the Section~\ref{sec:properties-of-symm} is devoted to the study
of Choi matrices given by self-adjoint unitaries in $n=3$ dimensional case.
Some final remarks are given in Section \ref{sec:fr}.

\section{Preliminaries}
\label{sec:preliminaries}

\subsection{Basic definitions and notation}
\label{sec:basic-notation}

By the $a\trpse$ we will denote as usual the transpose of the matrix $a$.
Occasionally when the function-of-argument notation will be more convenient we
will use $\transpose(a)$ to denote transposition map. By $\id$ we will denote
identity map. The $\circ$ will represent ordinary composition of maps. We
implicitly assume that all discussed maps between matrix algebras are linear.

Recall that the linear map $\phi\colon\malg{n}\to\malg{m}$ between the algebra
of $n\times n$ matrices and $m\times m$ matrices is called \emph{positive}
when it maps positive semidefinite matrices\footnote{We consider a matrix to
be positive semi-definite when its spectrum lies on positive half-line.}
(denoted by $\malg{n}^+$) into positive semidefinite matrices. We will denote
the set of all positive maps from $\malg{n}$ to $\malg{m}$ by
$\pmaps(\malg{n},\malg{m})$. We will write $\pmaps(\malg{n})$ instead of
$\pmaps(\malg{n},\malg{n})$. A map is called \emph{completely positive} when
maps $\phi\otimes \id\colon\malg{n}\otimes \malg{k}\to\malg{m}\otimes\malg{k}$
are positive for any $k$. We will denote by $ \cpmaps(\malg{n},\malg{m})$ the
set of all completely positive maps from $\malg n$ to $\malg m$. A map is
called \emph{completely copositive,} if $\transpose\circ\phi$ is a completely
positive map. A positive map is called \emph{decomposable} if it can be
written as a sum of completely positive and completely copositive map. It is
well known that in the case $\pmaps(\malg 2), \pmaps(\malg 2, \malg 3),
\pmaps(\malg 3, \malg 2)$ all maps are decomposable
\cite{woronowicz1976positive,choi1975completely}.

More generally, we can consider a positive and completely positive maps from a
\Cstar-algebra $\algebra A$ into algebra of bounded operators $\bmaps(\hilbert
H)$ acting on some Hilbert space $\hilbert H$. In that case we will denote by
$ \pmaps(\algebra A, \bmaps(\hilbert H))$ and $\cpmaps(\algebra
A,\bmaps(\hilbert H))$ sets of positive and completely positive maps
respectively. The case of $\malg{n}$ is a special case of this general
approach as matrix algebra is a finite dimensional \Cstar-algebra.
Another example, important in this article, is the set of positive maps from
the algebra of continuous complex-valued functions on a locally compact
Hausdorff space (this is canonical example of a commutative \Cstar-algebra) to some
matrix algebra. It is well known that for a commutative $\algebra A$ the sets
$\pmaps(\algebra A,\bmaps(\hilbert H))$ and $\cpmaps(\algebra
A,\bmaps(\hilbert H))$ are equal.

A linear map $\phi\colon\malg{n}\to\malg{m}$ is called \emph{unital} if
$\phi(\bbone_n) = \bbone_m$ ($\bbone_n$ denotes identity matrix in $\malg{n}$;
if the dimension will be clear from the context we will drop index $n$). Norm
of a map $\phi\colon\malg{n}\to\malg{m}$ is defined as usual, i.e.
$\norm{\phi} = \sup\set{\norm{\phi(a)}\setdef a\in\malg{n}, \norm{a} = 1}$,
where $\norm{a}$ for $a\in\malg{n}$ denotes operator norm. 

The set of all positive maps is a convex cone in the set of all linear
and continuous maps. The subset of normalized, unital positive maps is
a convex subset of the set of all positive maps. The subset of
normalized and unital completely positive maps is a convex subset of 
the set of normalized unital positive maps.

\subsection{Isomorphism between functionals and states}

The relation between mapping spaces and continuous bilinear forms on a
tensor products follows from the works of Grothendieck
\cite{grothendieck1955produits}. In the general setting it was already
known in 1960s (cf. \cite{wickstead1973linear}), and later was
reformulated in the linear algebra terms for finite dimensional case
by Choi and Jamiołkowski (\cite{choi1975completely} and
\cite{jamiokowski1972linear}) and now is widely known as
Choi-Jamiołkowski isomorphism. However, as the underlying geometry
will play the crucial role in the sequel we will use following
consequence of the Grothendieck construction.

\begin{lemma}[cf. \cite{stormer1986extension}]
  \label{thm:basic-lemma}
  There is an isometric isomorphism between
  $\maps(\malg{n},\malg{n})$ and bilinear forms in
  $(\malg{n}\otimes_\pi\malg{n})\conj$
  given by
  \begin{equation*}
    \tilde\phi \big(\sum_i^k a_i\otimes b_i\big) 
    = \sum_i^k \Tr \left(\phi(a_i)b_i\trpse\right).
  \end{equation*}
  Moreover, the map $\phi\in\pmaps(\malg{n},\malg{n})$ if and only if
  $\tilde\phi$ is positive on $\malg{n}^+\otimes_\pi\malg{n}^+$.
\end{lemma}

The $\malg{n}\otimes_\pi\malg{n}$ that appeared in the Lemma is by
definition a Banach space completion of an algebraic tensor product
in the projective norm given by
\begin{equation*}
  \pi(x) = \inf \Set{\sum_i^k\norm{a}\norm{b}_1\Setdef x=\sum_i^k
    a_i\otimes b_i, a_i\in\malg{n},b_i\in\malg{n}}
\end{equation*}
The norm $\norm{\cdot}_1$ denotes the trace norm,
i.e. $\norm{a}_1 = \Tr\abs{a}=\Tr(a\conj a)^{1/2}$.

As we work on finite dimensional spaces, we can represent the bilinear
form $\tilde\phi$ corresponding to a positive map $\phi$ by a density
matrix $\rho_\phi$ given by a well-known formula
\begin{equation}
  \label{eq:choi-isom}
  \rho_\phi = \sum_{ij} E_{ij} \otimes \phi(E_{ij}),
\end{equation}
where $E_{ij}$ are matrix units. The positivity condition from the
Lemma~\ref{thm:basic-lemma} can now be restated: a map $\phi$ is
positive if and only if corresponding $\rho_\phi$ is
\emph{block-positive}, what we denote by
\begin{equation*}
\rho_\phi\gebp0\quad \text{ iff }\quad  (x\otimes y, \rho_\phi x\otimes y)\ge 0, 
\quad \forall 
x,y\in\complexes^n.
\end{equation*}

Very important feature of the cited Lemma~\ref{thm:basic-lemma} is the
fact that it establishes isometric isomorphism, thus normed maps are
mapped into normed functionals. But as these functionals are defined
on projective tensor product, the corresponding functional norm
must be dual to the projective norm. We will denote this norm by
$\alpha$ and the duality tells us that
\begin{equation}
  \label{eq:alpha-norm}
  \alpha(\rho_\phi) = \sup\Set{\frac{\abs{\Tr \rho_\phi
        a}}{\pi(a)}\Setdef a\in\malg{n}\otimes_\pi\malg{n}, a\neq 0}.
\end{equation}
Using this we can specify the set of Choi matrices corresponding to
normalized positive maps
\begin{equation*}
  \setdzero \defeq \set{\rho\in\malg{n}\otimes_\pi\malg{n} \setdef
    \rho=\rho\conj, \alpha(\rho)=1, \rho\gebp0},
\end{equation*}
and the set Choi matrices corresponding to normalized, unital maps
(for a detailed justification see~\cite{Majewski:2010fk})
\begin{equation*}
  \setd \defeq \set{\rho\in\malg{n}\otimes\malg{n} \setdef
    \rho=\rho\conj, \alpha(\rho)=1, \rho\gebp0, \Tr \rho = n}.
\end{equation*}
As was mentioned in the Introduction, in the study of exposed points
of the set $\setd$ a distinguished role is played by selfadjoint
unitaries (see \cite{Majewski:2010fk}), thus we recall a definition.

\begin{definition}
 An operator $s$ is called a \emph{symmetry} if it is a selfadjoint unitary,
 i.e. $s=s\conj$ and $s^2=\bbone$. The set of all symmetries in the set
 $\bmaps(\hilbert H)$ of all bounded operators acting on Hilbert space
 $\hilbert H$ will be denoted by $\symmetries(\hilbert H)$. 

 An operator $s$ is called a \emph{partial symmetry} or $e$-symmetry if $s$ is
 selfadjoint and $s^2 = e$, where $e$ is some orthogonal projector on
 $\hilbert H$. 
\end{definition}

Note that any symmetry admits a canonical decomposition $s=p-q$, where
$p,q$ are orthogonal projectors such that $p+q=\bbone$. Namely
$p=1/2(\bbone+s)$ and $q=1/2(\bbone-s)$. In particular we can write
$s=\bbone - 2q$. Symmetries are also useful in computing the
$\alpha$-norm, as we see in following lemma.

\begin{lemma}[\cite{Majewski:2010fk}, Lemma 16]
  \label{thm:alpha-norm-lem}
  Let $\sigma\in\malg n\otimes_\alpha\malg n$, then
  \begin{equation*}
    \alpha(\rho) = \max\scaled\set{\abs{\Tr \rho s\otimes
        p}\setdef s\in\symmetries(\complexes^n),
      p\in\operatorname{Proj}^1(\complexes^n)},
  \end{equation*}
 where $\operatorname{Proj}^1(\complexes^n)$ stands for the set of rank one
 orthogonal projectors on Hilbert space $\complexes^n$. 
\end{lemma}

\subsection{Notions of extremality}
\label{sec:notions-extremality}

In the setting of positive maps different notions of extremality
arise. The most obvious is a notion of extreme point of a convex set. The map
$\phi\in\pmaps(\algebra A,\bmaps(\hilbert H))$ is called \emph{extreme} when it
cannot be written as a convex combination of other positive maps. 

Now consider a completely positive map $\phi\in\cpmaps(\algebra A,\bmaps(\hilbert H))$. One can write it in following way: $\phi = \sum_i t_i\conj \phi_i t_i$, where $t_i\in\bmaps(\hilbert H)$ such that $\sum t_i\conj t_i = \bbone$, all $t_i$ are invertible and all $\phi_i$ are completely positive (there is always a trivial decomposition). Such map is called \emph{\Cstar-extreme} in the set of completely positive maps whenever for all such decompositions all $\phi_i$ are unitarily equivalent to $\phi$ (for details see \cite{Gregg:2008lr}).

Finally, we can define an order structure in $\cpmaps(\algebra A,\bmaps(\hilbert H))$
in the following way \cite{arveson1969subalgebras}: $\psi\le\phi$ when
$\phi-\psi$ is completely positive. Then we call a completely positive map
$\phi$ \emph{pure} if $\psi\le\phi$ implies $\psi=\lambda \phi$ (thus it is
natural generalization of the notion of a pure state).

In general, for completely positive maps following inclusions are valid
\begin{equation*}
  \text{pure maps} \subseteq \text{\Cstar-extreme maps} \subseteq
  \text{extreme maps}
\end{equation*}
In many cases it is known that some of these inclusions are proper. For our
case it will be important to note that in general for maps $\pmaps(C(X), \malg
n)$ there are extreme maps that are not \Cstar-extreme. In the
section~\ref{sec:extr-points-pmapsd} we will reexamine this problem in the
special case of $\pmaps(C(X), \malg 2)$ using the Arveson
characterization~\cite{arveson1969subalgebras} of extreme maps in
$\cpmaps(C(X),\malg n)$.

\begin{dfn}[\cite{arveson1969subalgebras}]
 A family of subspaces $\set{\hilbert M_1,\dots \hilbert M_n}$ of Hilbert space
 $\hilbert H$ is \emph{weakly independent} if whenever there are given
 $\set{T_i\in\bmaps(\hilbert H)}_1^n$, such that the range of $T_i$ and
 $T_i\conj$ lies in $\hilbert M_i$, equality $T_1+\dots+T_n=0$ implies that
 $T_1=\dots=T_n=0$.
\end{dfn}

\begin{remark}[see \cite{arveson1969subalgebras}]
This condition is equivalent to a linear
independence of the family of subspaces $\set{\hilbert N_1,\dots
  \hilbert N_n}$ of $\hilbert H\otimes\hilbert H$, where $\hilbert
N_i\defeq [\xi\otimes\eta\setdef \xi,\eta\in\hilbert M_i]$.
\end{remark}

\begin{theorem}[\cite{arveson1969subalgebras}, Thm 1.4.10 
and also cf. \cite{Stormer:1963lr}]
 \label{thm:arv} 
 Let $X$ be a compact Hausdorff space and let $\hilbert H$ be a
 finite dimensional Hilbert space. Then extreme points of the set of
 unital completely positive maps $C(X)\to\malg{n}$ are maps of the
 form
  \begin{equation*}
    \phi(f) = f(x_1) K_1+\dots+f(x_k) K_k,\quad f\in C(X),
  \end{equation*}
  where $k\ge 1$, $x_1,\dots,x_k$ are distinct points of $X$ and
  $K_1,\dots,K_k$ are positive operators satisfying
  \begin{enumerate}
  \item[(i)] $K_1+\dots+K_k = \bbone$,
  \item[(ii)] $\set{[K_1\complexes^n], \dots, [K_k\complexes^n]}$ is weakly independent family of subspaces, where $[\hilbert h]$ denotes smallest subspace containing subset $\hilbert h$.
  \end{enumerate}
\end{theorem}

Note that any \Cstar-extreme map $\phi$ in
$\cpmaps(C(X),\complexes^n)$ is also extreme, thus can be
represented in the way showed in the previous theorem. Farenick and
Morenz \cite{farenick1997c} showed that the extreme
$\phi\in\cpmaps(C(X),\complexes^n)$ is \Cstar-extreme if and
only if $K_i$ are orthogonal projectors. This equivalently means that
$\phi$ is multiplicative.

We will relate the commutative case to the noncommutative case of
$\maps(\malg{n})$ by considering the restriction of a positive map
$\phi\in\pmaps(\malg{n})$ to the abelian subalgebra $\diagonal{n}\defeq \set{
a\in\malg{n}\setdef a \text{ is diagonal matrix}}$ of diagonal matrices. It is
well known fact that $a\in \diag n$ can be identified with $a_f\in C(X)$,
$X=\set{1,\ldots,n}$, i.e. the complex valued (trivially) continuous function
on the set $X$. Thus $\cpmaps(\diagonal{n},\malg{n})$ can be identified with
$\cpmaps(C(X),\malg n)$. We introduce following notion.

\begin{definition}
  \label{dfn:regpos} Let $\phi$ be a linear, extreme map in the set
  $\pmaps(\malg{n})$ or $\cpmaps(\malg{n})$. If the map $\phi_0\defeq
  \phi\restr{\diagonal{n}}$ is extreme in the set of
  $\cpmaps(\diagonal{n},\malg{n})$ then we call $\phi$ a \emph{regular extreme
  positive map.}
\end{definition}

\section{Extremality vs. \Cstar-extremality in abelian case}
\label{sec:extr-points-pmapsd}

Firstly, let us consider a special case of
$\cpmaps(C(X),\malg 2)$, where $X=\set{1,2}$. For $\phi$ extreme
in unital $\cpmaps(C(X),\malg 2)$ we conclude from Theorem
\ref{thm:arv} that
\begin{equation*}
  \phi(f) = f(x_0) K_0 \text{ or } \phi(f) = f(1) K_1 + f(2) K_2.
\end{equation*}
The first case implies that $K_0=\bbone$, so corresponding $\hilbert
M_0\defeq [K_0\hilbert H] = \hilbert H = \complexes^2$. Thus any
map of this form is also \Cstar-extreme.

Now take a closer look into the second case.  Let $e_1$ and $e_2$ are
unit vectors corresponding to projections onto $\hilbert M_1 =
[K_1\hilbert H]$ and $\hilbert M_2 = [K_2\hilbert H]$ respectively. As
$\hilbert M_i$ are weakly independent subspaces $\hilbert N_i =
[e_i\otimes e_i]$, $i=1,2$ are linearly independent. Moreover we know
that $K_i$ are positive and rank one operators. But any rank one
operator can be written in the form $\ket x\bra y$, and such operator
is hermitian if and only if $x=y$. So $K_i = \ket{x_i}\bra{x_i}$. But
\begin{equation*}
  \bbone = K_1+K_2 = \ket{x_1}\bra{x_1} + \ket{x_2}\bra{x_2}.
\end{equation*}
Now acting on $x_1$ on the right and taking scalar multiplication from the
left by $x_1$ we get that $(x_1,x_2)=0$ so $K_i$ are orthogonal
projectors. Thus by the Farenick and Morenz result any such map is
also \Cstar-extreme and therefore multiplicative (for details see
\cite{farenick1997c} and \cite{Gregg:2008lr}). As a result we proved
the following.

\begin{lemma}
  Any extreme map in $\cpmaps(C(\set{1,2}), \complexes^2)$ is \Cstar-extreme.
\end{lemma}

\bigskip

We will now discuss a more complicated case of $\pmaps(\diag 3,
\malg 3)$. Here, using the Arveson
Theorem~\cite{arveson1969subalgebras} we conclude that the dimensions
of $\hilbert M_i$ can be equal to $1,2,3$. The case of dimension 3 is
trivial, as before. Let us then consider the case when one of
$\hilbert M_i$'s have the dimension equal to 2.

\begin{example}
  \label{thm:3dcex}
  Take
  \begin{align*}
    K_1 &= \ket{e_1}\bra{e_1},\\
    K_2 &= \ket{e_2}\bra{e_2},\\
    K_3 &= \frac12\ket{e_1+e_3}\bra{e_1+e_3}+\ket{e_3}\bra{e_3}.
  \end{align*}
  Then $K_1+K_2+K_3=\bbone + P\equiv S$, where
  $P=\frac12\ket{e_1+e_3}\bra{e_1+e_3}$. Note that $S$ is invertible
  thus we can define
  \begin{equation*}
    \tilde K_i = S^{-\frac12} K_i S^{-\frac12}.
  \end{equation*}
  This does not change the rank of $K_i$ as $S^{-\frac12}$ is a
  non-singular matrix. It is also self-adjoint, thus this operation
  preserves positivity. Moreover $\tilde K_1+\tilde K_2+\tilde
  K_3=\bbone$. Thus we can define an extreme map
  $\cpmaps(C(X),\malg 3)$ by
  \begin{equation*}
    \tilde \phi (f) =  f(x_1)\tilde K_1+f(x_2) \tilde K_2 + f(x_3) \tilde K_3.
  \end{equation*}
  But using the matrix representation we compute that
  \begin{equation*}
    \tilde K_1 \tilde K_3 = \left(
      \begin{array}{ccc}
        \frac{1}{72} \left(5+2 \sqrt{6}\right) & 0 & -\frac{1}{72} \\
        0 & 0 & 0 \\
        -\frac{1}{72} & 0 & \frac{1}{72} \left(5-2 \sqrt{6}\right) \\
      \end{array}
    \right)
  \end{equation*}
  thus $K_1$ and $K_3$ are not orthogonal and by remark following the
  Theorem \ref{thm:arv} we conclude that the map $\tilde \phi$ is not
  multiplicative, so it is not \Cstar-extreme.
\end{example}
 
This example indeed shows that the set of \Cstar-extreme maps in $\cpmaps(\diag
3, \malg 3)$ is indeed a proper subset of the set of all extreme maps. It is
not surprising as even in $\cpmaps(C(\set{1,2,3}),\malg 2)$ there are examples
of extreme maps that are not \Cstar-multiplicative \cite{Gregg:2008lr}.

\section{Extreme positive maps on $2\times2$ vs. $3\times3$ matrices}
\label{sec:extr-posit-maps}

The results from the previous section allows us to get deeper insight
into the structure of the well known case of $\pmaps(\malg 2)$, as
well as understand a bit more the nature of qualitative change when we
increase the dimension by 1.

Fix a normalized unital $\phi\in \pmaps(\malg 2)$. Using the formula
\eqref{eq:choi-isom} we introduce following notation:
\begin{equation*}
  \rho_{\phi} = \sum_{ij} E_{ij}\otimes \phi(E_{ij}) = \sum_{ij}
  E_{ij}\otimes \rho_{ij},\quad\text{where } E_{ij} = \ket{e_i}\bra{e_j}
\end{equation*}
From the definition of $\rho_{ij}$ we immediately get that $\rho_{11}\ge0$,
$\rho_{22}\ge0$, $\rho_{11}+\rho_{22} = \bbone$ and $\rho_{ij} =
\rho_{ji}\conj$. In two dimensional case the structure of $\rho$ can be
explicitly given, namely:

\begin{proposition}
\label{thm:form-of-rho}
  The Choi matrix $\rho_\phi$ corresponding to the regular extreme normalized
  unital map
  $\phi\in\pmaps(\malg 2)$ can be written in one of following block
  forms in some matrix representation
  \begin{equation*}
    \rho_\phi = \left(
    \begin{matrix}
      \ket{y_1}\bra{y_1} & 
      c_0\ket{y_1}\bra{y_2}+c\ket{y_2}\bra{y_1} \\
      c_0\ket{y_2}\bra{y_1}+\cconj c\ket{y_1}\bra{y_2} &
      \ket{y_2}\bra{y_2}
    \end{matrix}\right)
    \text{ or }
    \rho_\phi = \left(
      \begin{matrix}
        \bbone & 0 \\
        0 & 0
      \end{matrix}\right),
  \end{equation*}
  where $c_0\ge0$, $c\in\complexes$ and $\set{y_1,y_2}$ is some basis
  in $\complexes^2$.
\end{proposition}

\begin{proof}
  If we consider a restriction of the map $\phi$ to diagonal matrices,
  then based on results of the previous section and the
  Def.~\ref{dfn:regpos}, we conclude that
  either
  \begin{equation*}
    \rho_{ii} = \ket{y_i}\bra{y_i}
  \end{equation*}
  or
  \begin{equation*}
    \rho_{11}=\bbone,\quad\rho_{22}=0.
  \end{equation*}
  Firstly we will consider non-trivial case. Assume $\rho_{ii} =
  \ket{y_i}\bra{y_i}$. The block-positivity property of the $\rho$
  gives us
  \begin{equation*}
    (x\otimes y, \sum_{ij} (E_{ij}\otimes \rho_{ij}) x\otimes y) = 
    \sum_{ij} (x,e_i)(e_j,x) (y,\rho_{ij} y)\ge0.
  \end{equation*}
  Now let us take $x = \epsilon e_1 + \lambda e_2$ and $y=y_1$, with
  $\epsilon>0$ and $\lambda$ real. Note that $x$ does not have to be
  normalized vector. Then the inequality above gives us
  \begin{equation*}
    \epsilon^2 + \epsilon\lambda (y_1,\rho_{12} y_1) +
    \epsilon \lambda (y_1,\rho_{21} y_1) + 0 \ge0 
  \end{equation*}
  So we get
  \begin{equation*}
    \lambda\big(y_1,(\rho_{12}+\rho_{12}\conj)y_1\big) \ge -\epsilon.
  \end{equation*}
  As vector $x$ can be chosen arbitrary, we can also take the vector
  $\epsilon e_1 - \lambda e_2$ and then we get
  \begin{equation*}
    \epsilon\ge \lambda\big(y_1,(\rho_{12}+\rho_{12}\conj)y_1\big)
  \end{equation*}
  Fixing $\epsilon$ and taking arbitrary $\lambda$ we conclude that
  \begin{equation}
    \big(y_1,(\rho_{12}+\rho_{12}\conj)y_1\big) = 0
  \end{equation}
  If we proceed by the same way using vectors $\epsilon e_1 + i
  \lambda e_2$ and $\epsilon e_1 - i \lambda e_2$ we conclude that
  \begin{equation}
    \label{eq:2}
    \big(y_1,(\rho_{12}-\rho_{12}\conj)y_1\big) = 0  
  \end{equation}
  Combining these two we get that $(y_1,\rho_{12} y_1) = 0$. If we
  choose $y=y_2$ and repeat all the above reasoning we arrive to
  conclusion that $(y_2, \rho_{12} y_2) = 0$, so finally we get that
  \begin{equation}
    \label{eq:3}
    \rho_{12} = c_1 \ket{y_1}\bra{y_2}+c_2\ket{y_2}\bra{y_1}.
  \end{equation}
  Now we do a unitary transformation $y_1\mapsto e^{-i \arg c_1}
  y_1$, $y_2\mapsto y_2$, which gives us the desired result. 

  In the case when
  $\rho_{11}=\bbone$ and $\rho_{22}=0$ we repeat all above
  calculations with the only difference that
  we get $(y,\rho_{12} y)=0$ for any $y$. Thus $\rho_{12}=0$ and
  this corresponds to the second form. 
\end{proof}

Let us now discuss admissible values of coefficients $c_0$ and
$c$. In this part we will extensively use the fact, than $\rho$ is
normalized in $\alpha$-norm, i.e. $\alpha(\rho_\phi)=1$. In
particular, the definition of $\alpha$-norm tells us that
\begin{equation*}
  1 = \alpha(\rho_\phi) \ge \abs{\Tr \rho_\phi a\otimes
    b},
\end{equation*}
for any $a$ and $b$ such that $\pi(a\otimes b) = 1$. Note that if
$\norm{a}=1$ and $\norm{b}_1$ then $\pi(a\otimes b)=1$. Take for
$a=E_{12}+\lambda E_{21}$ and $b=\ket{y_2}\bra{y_1}$, with
$\abs{\lambda}=1$. From the definition of the operator norm one
instantly gets that $\norm{a}=1$. On the other hand $\norm{b}_1 = \Tr
\big|\ket{y_2}\bra{y_1}\big| = \Tr (\ket{y_1}\bra{y_2}
\ket{y_2}\bra{y_1})^{1/2} = \Tr \ket{y_1}\bra{y_1} = 1$. Consequently
\begin{equation*}
  1 \ge \big|\Tr \left(\rho_\phi (E_{12}+\lambda E_{21})\otimes 
  \ket{y_2}\bra{y_1}\right)\big| = 
  \big| \Tr \left(\phi(E_{12}+\lambda E_{21}) \tau(\ket{y_2}\bra{y_1})\right)
\big|,
\end{equation*}
due to definition of $\rho_\phi$ (cf. Lemma \ref{thm:basic-lemma}). Now
applying the Proposition~\ref{thm:form-of-rho} we get
\begin{align*}
  1 &\ge \big| \Tr \left(\left(c_0\ket{y_1}\bra{y_2} + c\ket{y_2}\bra{y_1} +
  \lambda c_0\ket{y_2}\bra{y_1} + \lambda\cconj c\ket{y_1}\bra{y_2}
\right)\ket{y_1}\bra{y_2}\right)\big| \\
 &= \big| \Tr \left(c\ket{y_2}\bra{y_2}+\lambda c_0 \ket{y_2}\bra{y_2}\right) 
\big|
\end{align*}
Calculating the trace one arrives to $1\ge\abs{c+\lambda
  c_0}$. Because $\lambda$ here is arbitrary complex number of modulus
1, we can take in particular $\lambda=e^{i\arg c}$. Thus
\begin{equation}
  \label{eq:coeff-c-ineq}
  1\ge \big|\abs{c_0}+\abs{c}\big| = c_0+\abs{c}.
\end{equation}
Now it is easy to show that

\begin{theorem}
  Any regular extreme normalized unital map in $\pmaps(\malg2)$ corresponds to
  an element of the following subset of $\setd$
  \begin{equation*}
    \set{ \symmetries(\complexes^2), 2 P_x, p\otimes\bbone,
    \rho_{\tilde \phi} } = \tilde\setd\cup\set{\rho_{\tilde \phi}},
  \text{ where } 
  \rho_{\tilde\phi} = \left(
    \begin{matrix}
      \ket{y_1}\bra{y_1} & 0\\
      0 & \ket{y_2}\bra{y_2}
    \end{matrix}
  \right),
  \end{equation*}
  and $x$ is a maximally entangled vector in some basis
  $\set{y_1,y_2}$ of $\complexes^2$ and $p$ is a rank one projector in
  $\complexes^2$.
\end{theorem}
\begin{proof}
  If $\phi\in\pmaps{\malg2}$ is regular extreme, then from
  Proposition~\ref{thm:form-of-rho} we know that $\rho_\phi$ it is of the form
  \begin{equation*}
    \rho_\phi = \left(
    \begin{matrix}
      \ket{y_1}\bra{y_1} & 
      c_0\ket{y_1}\bra{y_2}+c\ket{y_2}\bra{y_1} \\
      c_0\ket{y_2}\bra{y_1}+\cconj c\ket{y_1}\bra{y_2} &
      \ket{y_2}\bra{y_2}
    \end{matrix}\right)
    \text{ or }
    \rho_\phi = \left(
      \begin{matrix}
        \bbone & 0 \\
        0 & 0
      \end{matrix}\right),
  \end{equation*}
  Let us consider the first case. Then $\set{y_1,y_2}$ fix basis in
  one Hilbert space. Let us use the same symbol to denote basis in the
  second Hilbert space (that is fixed by the matrix representation). Define
  $\tilde y_1 = y_1, \tilde y_2 = e^{i \arg c}y_2$ and
  \begin{equation*}
    w = \sum_{i,j=1}^{2} \ket{\tilde y_i}\bra{\tilde y_j} 
    \otimes\ket{y_j}\bra{y_i} = 
    \left(
      \begin{matrix}
        \ket{y_1}\bra{y_1} & e^{i \arg c}\ket{y_2}\bra{y_1} \\
        e^{-i \arg c}\ket{y_1}\bra{y_2} & \ket{y_2}\bra{y_2}
      \end{matrix}
    \right).
  \end{equation*}
  By straightforward calculation we check that $w$ is a symmetry. Now
  take $x=1/\sqrt 2(y_1\otimes y_1+y_2\otimes y_2)$ and define
  \begin{equation*}
    \rho_0 = 2 P_x = \left(
      \begin{matrix}
        \ket{y_1}\bra{y_1} & \ket{y_1}\bra{y_2} \\
        \ket{y_2}\bra{y_1} & \ket{y_2}\bra{y_2}        
      \end{matrix}
      \right).
  \end{equation*}
  Then one gets that
  \begin{equation*}
    \rho_\phi = c_0 \rho_0 + \abs c w + (1-c_0-\abs
    c)\rho_{\tilde \phi}.
  \end{equation*}
  Due to \eqref{eq:coeff-c-ineq} we see that $\rho_{\phi}$ must be a convex
  combination of maps form $\tilde\setd$ and $\rho_{\tilde \phi}$. As we
  assumed that $\phi$ is extreme, the claim follows.

  In the case when
  \begin{equation*}
        \rho_\phi = \left(
      \begin{matrix}
        \bbone & 0 \\
        0 & 0
      \end{matrix}\right),
  \end{equation*}
  we immediately see that it is equal to $\ket{y_1}\bra{y_1}\otimes
  \bbone$ in the basis fixed by matrix representation.
\end{proof}

\begin{remark}
  The element $\rho_{\tilde \phi}$ corresponds to the map projecting element $a$
  onto the subalgebra of diagonal matrices in some basis fixed by
  matrix representation. Namely
  \begin{equation*}
    \tilde \phi(a) = \left(
      \begin{matrix}
        a_{11} & 0 \\
        0 & a_{22}
      \end{matrix}
    \right).
  \end{equation*}
\end{remark}

\begin{remark}
It is also noteworthy to mention that maps corresponding to elements $2
P_x$ are isomorphisms and those corresponding to symmetries are
anti-isomorphisms. The last claim follows from the fact that in the $n=2$ case
all symmetries in $\setd$ are locally unitary equivalent to the Choi matrix of
transposition map (for a simple proof see \cite{Majewski:2010fk}). 
\end{remark}

The situation in the case of $\phi\in\pmaps(\malg 3)$ is much more complicated.
Our results concerns only regular maps. Nevertheless example~\ref{thm:3dcex}
shows that even in this case we cannot infer that the block-diagonal part of
Choi matrix, i.e. elements $\phi(e_{ii})$, are formed by the orthogonal
projectors. Moreover for $n=3$ there appear non-decomposable maps.
Illustration of this fact is given by generalized Choi maps.

\begin{example}
  \label{ex:choi}
  Consider a generalized Choi map of the form \cite{cho1992generalized, choi1975completely, tanahashi1988indecomposable}
  \begin{equation*}
    \phi(a) = \frac12\left(
    \begin{matrix}
       {a}_{11}+{a}_{33} & -{a}_{1,2} & -{a}_{1,3} \\
       -{a}_{2,1} & {a}_{22}+{a}_{11} & -{a}_{2,3} \\
       -{a}_{3,1} & -{a}_{3,2} & {a}_{33}+{a}_{22}
    \end{matrix}\right).
  \end{equation*}
  It is known that this is an extreme
  positive map. Arverson's decomposition of its restriction to
  commutative algebra $\diag 3$ is given by
  \begin{equation*}
    \phi(f) = f(x_1) K_1 + f(x_2) K_2 + f(x_3) K_3,
  \end{equation*}
  where
  \begin{align*}
    K_1 &= 1/2(\ket{e_1}\bra{e_1}+\ket{e_2}\bra{e_2})\\
    K_2 &= 1/2(\ket{e_2}\bra{e_2}+\ket{e_3}\bra{e_3})\\
    K_3 &= 1/2(\ket{e_1}\bra{e_1}+\ket{e_3}\bra{e_3}).
  \end{align*}
  It is apparent that in this example $K_1K_2 \neq 0$. 
\end{example}

\section{Convex analysis of Choi map} 
\label{sec:choi_map}

In this section we will study the structure of Choi map $\phi$ that was
recalled in the Example~\ref{ex:choi}. It is worth remembering that this was
the first and very important example of a non-decomposable positive map. Denote
by $\rho_C$ a Choi matrix corresponding to $\phi$ and by $\tilde \rho_C$ partial
transpose of $\rho_C$. From \cite{Majewski:2010fk} we know that partial
transposition preserves the set $\setd$. In fact $\tilde \rho_C$ corresponds to
the map $\tau \circ \phi$ and is also extreme and indecomposable. The analysis
of $\tilde \rho_C$ is nicer that $\rho_C$. Thus to understand the nature of the
Choi map we will carry out an examination of $\tilde \rho_C.$

\begin{lemma}
  Let
  \begin{equation*}
    w^- = \sum_{i,j=1}^3 \varepsilon_{ij} E_{ij}\otimes E_{ji},\quad \text{where }
    \varepsilon_{ij} = \begin{cases}
    1 & \text{for } i = j,\\
    -1 & \text{for } i\neq j.
    \end{cases}
  \end{equation*}
  Then $w^-$ is a symmetry (but not block positive) and $\alpha(w^-) = 5/3.$ 
\end{lemma}
\begin{proof}
  The fact that $w^-$ is a symmetry follows from the direct calculation. To see that it is not block positive it is enough to consider $x=1/2(e_1+e_3)+1/\sqrt2 e_2$ and calculate that
  \begin{equation*}
    (x\otimes, w^- x\otimes x) = -\frac14.
  \end{equation*}
  In order to calculate $\alpha(w^-)$ we
  will use Lemma~\ref{thm:alpha-norm-lem}, i.e.
  \[  
    \alpha(w^-) = \sup_{s,p} \abs{\Tr w^- s\otimes p},
  \]
 where $s$ is a symmetry and $p$ is a rank 1 projector. Because $s\in\malg 3$,
 we can write it as $s = \bbone - 2 q$, where $q$ is projector. Without loss
 of generality we can assume that this is rank 1 projector, as case rank 0 is
 trivial and rank 2 can be reduced to rank 1 by $\bbone-2q = 2q'-\bbone =
 -(\bbone - 2q')$ where $q' = \bbone - q$, and $q'$ is rank 1. Thus
  \[
    \alpha(w^-) = \sup_{q,p} \abs{\Tr w^- \bbone\otimes p - 2\Tr w^- q\otimes p}.
  \]
  Let $p=\ket{x}\bra{x}$ and $q=\ket{y}\bra{y}$. By explicit calculation we see that $\Tr w^- \bbone\otimes p = \norm{x}^2 = 1$, so to obtain supremum we need to find extreme values of $\Tr w^- q\otimes p$. Denote by $\set{x_i}$ and $\set{y_i}$ coefficients of $x$ and $y$ in canonical basis. Then we calculate
  \begin{equation*}
    \begin{split}
     \Tr w^- q\otimes p &= \abs{x_1}^2\abs{y_1}^2 - x_1 \cconj x_2 y_2 \cconj y_1 - x_1 \cconj x_3 y_3 \cconj y_1 \\
     &- x_2 \cconj x_1 y_1\cconj y_2 + \abs{x_2}^2\abs{y_2}^2 - x_2 \cconj x_3 y_3 \cconj y_2 \\
     &- x_3\cconj x_1 y_1 \cconj y_3 - x_3 \cconj x_2 y_2 \cconj y_3 + \abs{x_3}^2\abs{y_3}^2 \\
     &= \abs{x_1}^2\abs{y_1}^2 + \abs{x_2}^2\abs{y_2}^2 + \abs{x_3}^2\abs{y_3}^2\\
     &- 2 \Re x_1 \cconj x_2 y_2 \cconj y_1 
     - 2 \Re x_1 \cconj x_3 y_3 \cconj y_1
     - 2 \Re x_2 \cconj x_3 y_3 \cconj y_2.
     \end{split} 
  \end{equation*}
  We can rewrite this using a polar decomposition of complex coefficients $x_j = \xi_j e^{i \phi_j}, y_j = \eta_j e^{i \psi_j}$
  \begin{equation*}
    \begin{split}
      \Tr w^- q\otimes p &= \xi_1^2 \eta_1^2 + \xi_2^2 \eta_2^2 + \xi_3^2 \eta_3^2 \\
      &- 2 \xi_1 \xi_2 \eta_1 \eta_2 \cos (\phi_1 - \phi_2 + \psi_2 - \psi_1) \\
      &- 2 \xi_1 \xi_3 \eta_1 \eta_3 \cos (\phi_1 - \phi_3 + \psi_3 - \psi_1) \\
      &- 2 \xi_2 \xi_3 \eta_2 \eta_3 \cos (\phi_2 - \phi_3 + \psi_3 - \psi_2) \\
      & \ge \xi_1^2 \eta_1^2 + \xi_2^2 \eta_2^2 + \xi_3^2 \eta_3^2
      - 2 \xi_1 \xi_2 \eta_1 \eta_2 
      - 2 \xi_1 \xi_3 \eta_1 \eta_3 
      - 2 \xi_2 \xi_3 \eta_2 \eta_3 = m,
    \end{split}
  \end{equation*}
  with equality e.g. for $\phi_i = 0 = \psi_j$, and
  \begin{equation*}
    M = \xi_1^2 \eta_1^2 + \xi_2^2 \eta_2^2 + \xi_3^2 \eta_3^2
      + 2 \xi_1 \xi_2 \eta_1 \eta_2 
      + 2 \xi_1 \xi_3 \eta_1 \eta_3 
      + 2 \xi_2 \xi_3 \eta_2 \eta_3 \ge \Tr w^- q\otimes p.
  \end{equation*}
  with equality e.g. for $\phi_1 = \phi_2 = \pi$ and other $\phi_i = 0, \psi_i = 0$. 
  We use the normalization of $x$ and $y$ to introduce parametrization
  \begin{align*}
    \xi_1 &= \sin \alpha \sin \beta, &\quad \xi_2 &= \cos \alpha \sin \beta, &\quad \xi_3 &= \cos \beta,\\
    \eta_1 &= \sin \mu \sin \nu, &\quad \eta_2 &= \cos \mu \sin \nu, &\quad \eta_3 &= \cos \nu,
  \end{align*}
  Substitution and simplification yields
  \begin{equation*}
    \begin{split}
    m &=\cos^2 \beta \cos^2 \nu + \cos^2 (\alpha + \mu) \sin^2 \beta \sin^2 \nu - \cos \beta \sin \beta \cos (\alpha - \mu) \sin 2 \nu \\
    &= (\cos \beta \cos \nu - \cos (\alpha + \mu) \sin \beta \sin \nu)^2 - 4 \cos \beta \sin \beta \sin \alpha \sin \mu \sin \nu \cos \nu,
    \end{split}
  \end{equation*}
  and
  \begin{equation*}
    M = (\cos \beta \cos \nu + \cos (\alpha - \mu) \sin \beta \sin \nu)^2.
  \end{equation*}  
  Now we substitute
  \begin{align*}
    \alpha^- &= \alpha - \mu, & \alpha^+ &= \alpha + \mu,\\
    \beta^- &= \beta - \nu, & \beta^+ &= \beta + \nu
  \end{align*}
  and get 
  \begin{gather}
  \begin{split}
    m &=\frac14 (\cos \beta^- + \cos \beta^+ - \cos \alpha^+ \cos \beta^- + \cos \alpha^+ \cos \beta^+)^2 \\
    &- \frac12 (\cos \beta^- + \cos \beta^+)(\cos \beta^- - \cos \beta^+)(\cos \alpha^- -  \cos \alpha^+)
  \end{split}\\
    M = \frac14 (\cos \beta^- + \cos \beta^+ + \cos \alpha^- \cos \beta^- - \cos \alpha^- \cos \beta^+)^2.
  \end{gather}
  Finally we denote
  \begin{equation*}
    a = \cos \beta^-,\quad b = \cos \beta^+,\quad c = \cos \alpha^+,\quad d = \cos \alpha^-,
  \end{equation*}
  to get
  \begin{align*}
    m&=\frac14(a + b - c a + c b)^2 - \frac12 (a+b)(a-b)(d-c)\\
    M&=\frac14(a + b + d a - d b)^2
  \end{align*}
  By direct calculation the minimum value of $m$ is equal to $-1/3$ (e.g. for $a=-1,b=1/3,c=0,d=1$) and maximum value of $M$ equals $1$, so the $\alpha$-norm of $w^-$ equals $5/3$.
\end{proof}

\begin{remark}
  For $n=2$ analogously defined $w^-$ belongs to $\setd$. 
\end{remark}

Simple calculation leads to following conclusion.

\begin{proposition}
  Choi matrix $\tilde \rho_C$ is a convex combination of $w^-$ and matrix
  \begin{equation*}
    r = E_{11}\otimes E_{22} + E_{22}\otimes E_{33} + E_{33}\otimes E_{11},
  \end{equation*}
  namely
  \begin{equation}
    \label{eq:tilderho}
    \tilde \rho_C = \frac12 r + \frac12 w^-.
  \end{equation}
\end{proposition}

\begin{remark}
  Matrix $r$ is positive and it is straightforward to check that $r\in\setd$. 
\end{remark}

Let us recall the definition of generalized Choi map (see e.g.\cite{choi1973positive}):
\begin{equation*}
  \phi_{a,b,c}(x) = \psi_{a,b,c}(x) - x,
\end{equation*}
where
\begin{equation*}
  \psi_{a,b,c}(x) = \left(
    \begin{matrix}
      a x_{11} + b x_{22} + c x_{33} & 0 & 0 \\
      0 & a x_{22} + b x_{33} + c x_{11} & 0 \\
      0 & 0 & a x_{33} + b x_{11} + c x_{22}
    \end{matrix}
  \right).
\end{equation*}
Such map is positive if and only if following conditions are satisfied
\begin{enumerate}[(i)]
\item $a\ge 1$,
\item $a+b+c \ge 3$,
\item $bc\ge (2-a)^2$ if $1\le a \le 2$. 
\end{enumerate}
Now one can see that partial transpose of $w^-$ is a Choi matrix corresponding to generalized Choi map $\phi_{2,0,0}$.
Then for any $0\le \lambda\le 1$ 
\begin{equation*}
  \rho_{\lambda} = \lambda r + (1 - \lambda) w^-
\end{equation*}
corresponds to generalized Choi map $(1 - \lambda) \phi_{2,0,\lambda/(1 - \lambda)}$. The factor $(1 - \lambda)$ in front ensures that the map is always unital. Conditions under which the generalized Choi map is positive imply that for $\lambda \ge 1/2$ the Choi matrix $\rho_{\lambda}$ belongs to $\setd.$

\section{Properties of symmetries in $\mathfrak D$}
\label{sec:properties-of-symm}

To further examine the $n=3$ case we will focus on a local unitary equivalence
of those Choi matrices that are represented by symmetries. We know that in
$n=2$ case all symmetries in $\setd$ are locally unitarily equivalent to Choi
matrix representing transposition map. The natural question arise whether it
is still true in $n=3$ case or can symmetries also represent some
non-decomposable maps. 

Through this section we adopt convention that coefficients of Schmidt
decomposition are non-negative (any possible phase is included in vectors of
Schmidt decomposition). To simplify notation we used the same symbol $e_i$ to
denote basis vectors in the first and the second Hilbert space, but clearly
this is only a matter of convenience.

\subsection{Technical lemmas}
\label{sec:technical-lemmas}

\begin{lemma}
  \label{thm:rank2fullent}
  Let $s$ be a block positive symmetry, with decomposition
  $s=p-q$. Then 
  \begin{enumerate}
  \item[(i)] any eigenvector of $q$ must have Schmidt rank greater than one;
  \item[(ii)] any eigenvector of $q$ that have Schmidt rank equal to 2
    must have both Schmidt coefficients equal to $1/\sqrt{2}$.
  \end{enumerate}
\end{lemma}
\begin{proof}
  Block positivity condition implies that for normalized vectors
  $$ (x\otimes y, q x \otimes y) \le \frac{1}{2}.$$
  The (i) part is then obvious, as one could take for $x\otimes y$
  eigenvector of $q$ that have Schmidt rank equal to one and violate
  above inequality.

  For (ii) let us consider the Schmidt rank 2 normalized eigenvector $v$ of
  $q$. Its Schmidt decomposition can be written as $v=\cos\alpha\; e_1\otimes
  f_1+\sin\alpha\; e_2\otimes f_2$ with $\alpha\in(0,\pi/2)$. Then:
  \begin{align*}
    (e_1\otimes f_1, P_v e_1\otimes f_1) &= \cos^2 \alpha,\\
    (e_2\otimes f_2, P_v e_2\otimes f_2) &= \sin^2 \alpha,
  \end{align*}
  where $P_v$ denotes orthogonal projector on vector $v$. But when
  $\cos^2\alpha\le1/2$ then $\sin^2\alpha\ge1/2$ with equality only when both
  equal $1/2$. So Schmidt coefficients of $v$ must be equal to $1/\sqrt{2}$.
\end{proof}

Following lemma about $\alpha$-norm will be used very
often and we will use it without an explicit mention.

\begin{lemma}
  \label{thm:alpha-norm}
  Let $\rho\in\mathfrak D$. Then for any one-dimensional projector $p$ 
  \begin{equation*}
    \alpha(\rho) = \Tr (\bbone\otimes p)\rho = 1 
 \end{equation*}
\end{lemma}
\begin{proof}
  From the Lemma~\ref{thm:alpha-norm-lem} we have
  \begin{equation*}
    \alpha(\rho) = \max\scaled\set{\abs{\Tr \rho s\otimes
        p}\setdef s\in\symmetries(\hilbert
      H),p\in\operatorname{Proj}^1(\hilbert H)}
  \end{equation*}
  Because $\rho\in\mathfrak D$, then there exists positive and
  unital map $\phi_\rho\colon\maps(\hilbert H)\to\maps(\hilbert H)$.
  For any symmetry $s$ and projector $p$ we thus have
  \begin{equation*}
    \Tr \rho s\otimes p = \Tr \phi_\rho(s)p\trpse.
  \end{equation*}
  Next, we notice that from the Kadison inequality we have that
  \begin{equation*}
    \bbone = \phi_\rho(s^2) \ge \phi_\rho(s)^2,
  \end{equation*}
  thus $\bbone \ge \abs{\phi_\rho(s)}$ and by the Proposition
  2.2.13c in \cite{bratteli2003operator},
  \begin{equation*}
    p\trpse \ge p\trpse\abs{\phi_\rho(s)}p\trpse.
  \end{equation*}
  Let $p\trpse=\ket x\bra x$. By $P_s(\lambda)$ we denote spectral projections
  of $\phi_\rho(s)$. Taking trace we have
  \begin{align*}
    \Tr p\trpse &\ge \Tr p\trpse\abs{\phi_\rho(s)}p\trpse = (x, \abs{\phi_
\rho(s)} x)
    = \sum_{\lambda\in\sigma(\phi_\rho(s))} \abs{\lambda} (x,P_s(\lambda) x) 
\\
    &\ge \scaled\abs{\sum_{\lambda\in\sigma(\phi_\rho(s))} \lambda 
(x,P_s(\lambda) x)}
\unscale
    = \abs{ (x,\phi_\rho(s) x) } = \abs{\Tr \phi_\rho(s)p\trpse}  = 
    \abs{\Tr \rho s\otimes (p\trpse)\trpse} \\
    &= \abs{\Tr \rho s\otimes p}.
  \end{align*}
  On the other hand we have $1=\Tr p\trpse = \Tr \phi_\rho(\bbone)p\trpse =
  \Tr \rho \bbone\otimes p$, for any $p$. 
\end{proof}

The following Lemma is in fact valid for any finite dimensional Hilbert
space $\hilbert H$. Although it seems to be well known for sake of
completeness we give here the proof because this lemma is the crucial element of
many proofs in the sequel. 

\begin{lemma}
  \label{thm:schmidtcoeff}
  Let $x=\sum_i \lambda_i e_i\otimes f_i$ be the Schmidt decomposition
  of vector $x$. Then for any
  one-dimensional projector $p$
  \begin{equation*}
    \Tr (\bbone\otimes p) P_x \le \max_i \lambda_i^2.
  \end{equation*}
  Moreover $\Tr (\bbone\otimes P_z) P_x = \max_i \lambda_i^2$ if and only
  if 
  \[
    z\in\operatorname{span} \set{f_i\setdef \text{where $i$ is such that } \lambda_i = \max_k \lambda_k}.
  \] 
\end{lemma}
\begin{proof}
  Let $z=\sum_i^k z_i f_i$ (if $k<N$, where $N$ is a dimension of
  a corresponding Hilbert space then we define $f_i$ for
  $i=k+1,\dots,N$ as mutually orthonormal vectors to $f_i,
  i=1,\dots,k$, such that $\set{f_i}_{i=1,\dots N}$ is a basis). Then 
  \begin{equation*}
    \begin{split}
      \Tr (\bbone\otimes P_z)P_x &= \sum_{ij} \lambda_i\lambda_j \Tr
      (\bbone\otimes P_z) \ket{e_i\otimes f_i}\bra{e_j\otimes f_j} \\
      &= \sum_{ij,mn} \lambda_i\lambda_j z_m \cconj z_n \Tr
      (\bbone\otimes \ket{f_m}\bra{f_n})
      \ket{e_i}\bra{e_j}\otimes\ket{f_i}\bra{f_j}\\ &= \sum_{ij,mn}
      \lambda_i\lambda_j z_m \cconj z_n \delta_{ij}
      \delta_{ni}\delta_{mj} \\ &= \sum_{ij} \lambda_i\lambda_j z_j
      \cconj z_i\delta_{ij} = \sum_i \lambda_i^2\abs{z_i}^2.
    \end{split}
 \end{equation*}
 Denote $\lambda_{\max}=\max_i \lambda_i$. Then
 \begin{equation*}
   \sum_i \lambda_i^2\abs{z_i}^2\le \sum_i \lambda_{\max}^2\abs{z_i}^2
   = \lambda_{\max}^2.
 \end{equation*}
 For a second claim, let us assume that $\lambda_i$ are sorted and the first
 $n$ of $\lambda_i$ are equal $\lambda_{\max}$ ($n$ can be smaller than $N$,
 in particular $n$ can be equal 1). The ``if'' part is obvious: substitution
 of $\lambda_{\max}$ for $\lambda_i$ does not change anything. The ``only if''
 part follows from the fact, that if $z=\sum_{i=1}^n z_i f_i + z_{j} f_{j}$,
 $j>n$ (with possible $z_i=0$ for $i\in1\dots n$), then $\lambda_j^2
 \abs{z_j}^2<\lambda_{\max}^2 \abs{z_j}^2$.
\end{proof}

We adopt following notation for partial transpose: $\bbone\otimes \tau
\equiv \ptranspose$.

\begin{lemma}
\label{thm:taupad}
Let $a\in \bmaps(\hilbert H)\otimes \bmaps(\hilbert H)$ and $U,V$ are
unitaries acting on $\hilbert H$. Then
  \begin{equation*}
    (\ptranspose\circ\ad_{U\otimes V})a = (\ad_{U\otimes \tau(V\conj)}\circ
    \ptranspose)a
  \end{equation*}
\end{lemma}
\begin{proof}
  Let $a = \sum_i a_i\otimes b_i$ then 
  \begin{equation*}
    \begin{split}
      \ptranspose(U\otimes V a U\conj\otimes V\conj) &= \sum_i (U a_i
      U\conj)\otimes(\tau(V b_i V\conj)) \\&= \sum_i (U a_i
      U\conj)\otimes(\tau(V\conj)\tau(b_i)\tau(V)) \\&= 
      U\otimes\tau(V\conj)\left(\sum_i a_i\otimes \tau(b_i)\right) U\conj\otimes
\tau(V) \\&=
      \left(U\otimes\tau(V\conj)\right) \ptranspose(a) \left(U\conj\otimes
\tau(V)\right).
    \end{split}
  \end{equation*}
\end{proof}

\subsection{Building blocks of symmetries in $\mathfrak D$}
\label{sec:building-blocks}

\begin{proposition}
  Let $\xi_1,\xi_2,\xi_3$ be three orthonormal vectors in $\hilbert
  H\otimes \hilbert H$ with Schmidt decompositions of the form: 
 \begin{align*}
    \xi_1 &= \sum_{i=1}^3 \lambda_i e_i\otimes f_i,\quad \lambda_i>0,
    \sum_{i=1}^3 \lambda_i^2 = 1,\\
    \xi_2 &= \frac{1}{\sqrt{2}} (h_1\otimes g_1 + h_2\otimes g_2),\\
    \xi_3 &= \frac{1}{\sqrt{2}} (k_1\otimes l_1 + k_2\otimes l_2).
  \end{align*}
  Then the symmetry
  $s=\bbone - 2 q$, where $q=\sum_i P_{\xi_i}$ is not in $\mathfrak D$.
\end{proposition}
\begin{proof}
   Assume that Schmidt coefficients of $\xi_1$ are sorted and the
  greatest is $\lambda_1$. We will consider separately three cases
  exhausting all possible values for $\lambda_1$, namely
  $\lambda_1>1/\sqrt{2}, \lambda_1\in[1/\sqrt{3},1/\sqrt{2})$, and
  $\lambda_1=1/\sqrt{2}$ ($\lambda_1$ is the greatest Schmidt
  coefficient, so must be greater or equal to $1/\sqrt3$). Last two
  parts will be proved by contradiction: we assume that $s$ is in
  $\mathfrak D$ and show that then $\alpha(s)\neq 1$. 

  If $\lambda_1>1/\sqrt{2}$ then $(e_1\otimes f_1, q e_1\otimes
  f_1)>1/2$ so $s$ is not block positive.

  If $\lambda_1\in[1/\sqrt3,1/\sqrt2)$, take
  the $l_3$ such that $\set{l_i}$ is a basis in $\hilbert H$. Then $\Tr
  (\bbone\otimes P_{l_3}) P_{\xi_3} = 0$.  But from the
  Lemma~\ref{thm:schmidtcoeff}: $\Tr (\bbone\otimes
  P_{l_3}) P_{\xi_1}<1/2$ and $\Tr (\bbone\otimes
  P_{l_3})P_{\xi_2}\le1/2$. Consequently 
  \begin{equation*}   
    \Tr (\bbone\otimes P_{l_3}) s = \Tr (\bbone\otimes P_{l_3}) \bbone
    - 2 \Tr (\bbone\otimes P_{l_3}) q = 3 - 2 \Tr (\bbone\otimes
    P_{l_3}) q > 3-2(1/2+1/2+0)=1,
  \end{equation*}
  and $s$ can not be in $\mathfrak D$.
 
  It remains to consider the case when $\lambda_1 = 1/\sqrt2$. Notice
  that then $\lambda_2,\lambda_3<1/\sqrt{2}$.
  Let $g_3,h_3,k_3$ and $l_3$ be orthonormal vectors to, respectively,
  $\set{g_1,g_2}, \set{h_1,h_2}, \set{k_1,k_2}$ and
  $\set{l_1,l_2}$. Then either (a) $f_1 \neq g_3$ or (b) $f_1=g_3$. In
  the case (a), from the Lemma~\ref{thm:schmidtcoeff} we infer that
  \begin{equation*}    
    \Tr (\bbone\otimes P_{g_3}) P_{\xi_1} < \frac{1}{2}, 
  \end{equation*}
  as the maximal Schmidt coefficient equals to $1/\sqrt2$, and $g_3$
  does not belong to one dimensional subspace spanned by
  $f_1$. But as $\Tr (\bbone\otimes P_{g_3}) P_{\xi_2} = 0$ ($g_3$ is
  orthogonal to $g_1$ and $g_2$), and $\Tr (\bbone\otimes P_{g_3})
  P_{\xi_3}\le1/2$ (Lemma~\ref{thm:schmidtcoeff} again), we conclude that
  \begin{equation*}
    \Tr (\bbone\otimes P_{g_3}) s > 3 - 2(1/2+1/2) = 1,
  \end{equation*}
  and $s$ is not in $\mathfrak D$ this case. 

  If $f_1=g_3$, then we repeat previous
  reasoning, i.e. either (b1) $f_1\neq l_3$ or (b2) $f_1=l_3$. The
  case (b1) can be treated exactly in the same manner as it was done
  previously in the case
  (a): $\Tr (\bbone\otimes P_{l_3}) P_{\xi_1} < \frac{1}{2}$ and the
  rest follows as before, so $s\notin \mathfrak{D}$. 

  For (b2) we conclude that $g_3=l_3$ and notice that $\Tr (\bbone
  \otimes P_{g_1}) P_{\xi_2}=1/2$ (obvious) and $\Tr (\bbone\otimes
  P_{g_1}) P_{\xi_3}=1/2$ (again Lemma~\ref{thm:schmidtcoeff}, as
  $g_1$ being orthogonal to $l_3=g_3$ belongs to the subspace spanned
  by $l_1,l_2$) and $\Tr (\bbone\otimes P_{g_1}) P_{\xi_1}$ must be
  strictly greater than zero as $\set{f_1,f_2,f_3}$ spans whole
  $\hilbert H$, so
  \begin{equation*}
    \Tr (\bbone\otimes P_{g_1}) q = \Tr (\bbone\otimes P_{g_1})
    P_{\xi_1}+\Tr (\bbone\otimes P_{g_1}) P_{\xi_2}+ \Tr (\bbone\otimes
    P_{g_1}) P_{\xi_3} > 1,
  \end{equation*}
  and also it that case $\Tr (\bbone\otimes P_{g_1}) s\neq 1$. We
  excluded all possibilities, so such $s$ cannot be in $\mathfrak D$.
\end{proof}

In the tensor product $\hilbert H\otimes\hilbert H$ the subspace of Schmidt
rank 3 vectors is one dimensional
\cite{cubitt2007dimension,parthasarathy2004maximal}, thus it is impossible to
have two or more such orthonormal vectors. Thus we arrive at the following
conclusion.

\begin{corollary}
  \label{thm:schmidt-rank2}
  Let $s$ be a block positive symmetry in $\mathfrak D$, and
  $s=p-q$. Then all eigenvectors of $q$ are 
  Schmidt rank 2 vectors with both Schmidt coefficients equal to $1/\sqrt2$.
\end{corollary}

\subsection{Local unitary equivalence of a certain class of symmetries}
\label{sec:local-unit-equiv}

It will be less complicated if we show locally unitary
equivalence to the following symmetry in $\mathfrak D$ (which is
locally unitary equivalent to symmetry corresponding to transposition
map). 

\begin{lemma}
\label{thm:s0}
  Let 
  \begin{align*}    
    x_1 &= \frac{1}{\sqrt{2}} (e_1\otimes e_1 + e_2\otimes e_2),\\
    x_2 &= \frac{1}{\sqrt{2}} (e_1\otimes e_3 + e_3\otimes e_2),\\
    x_3 &= \frac{1}{\sqrt{2}} (e_2\otimes e_3 - e_3\otimes e_1).
  \end{align*}
  Then $s_0 = \bbone - 2 \sum_i P_{x_i}$ is a block positive symmetry
  in $\mathfrak D$ locally unitary equivalent to the symmetry
  corresponding to the transposition map.
\end{lemma}
\begin{proof}
  By direct calculation one sees that partial transpose of $s_0$, is
  equal to $3 P_x$, where $x$ is maximally entangled vector:
  \begin{equation*} 
    x = \frac{1}{\sqrt3}(e_3\otimes e_3-e_1\otimes e_2+e_2\otimes e_1).
  \end{equation*}

  Now let $w$ be the Choi matrix corresponding to the transposition map
  (in the basis introduced above). Let $y$ be the vector defined by $3
  P_y = \ptranspose(w)$, i.e.
  \begin{equation*}
    y = \frac1{\sqrt3} (e_1\otimes e_1 + e_2\otimes e_2 + e_3\otimes e_3)
  \end{equation*}
  We remind that if two vectors have exactly the same Schmidt coefficients then
  they are locally unitarily equivalent, so $x = U\otimes V y$ for some
  unitaries $U,V$. Consequently $P_x = U\otimes V P_y U\conj\otimes V\conj$.
  Finally, by the Lemma~\ref{thm:taupad}:
  \begin{equation*}
    s_0 = \ptranspose(3 P_x) = \ptranspose(U\otimes
    V 3 P_y U\conj\otimes V\conj) = U\otimes \tau(V\conj) \ptranspose(3 P_y)
    U\conj\otimes \tau(V) = \ad_{U\otimes \tau(V\conj)} w.
 \end{equation*}
\end{proof}

Now we can prove our first equivalence result.

\begin{proposition}
  \label{thm:twofixed}
  Let $s$ be a symmetry in $\mathfrak D$ and let $s=\bbone - 2 q$.
  Assume that eigenvectors of $q$ are of the form
  \begin{align*}
    x_1 &= \frac{1}{\sqrt{2}} (e_1\otimes e_1 + e_2\otimes e_2),\\
    x_2 &= \frac{1}{\sqrt{2}} (e_i\otimes e_j \pm e_k\otimes
    e_l),\\
    x_3 &\text{ arbitrary consistent with assumptions,} 
  \end{align*}
  then $s$ is locally unitary equivalent to the symmetry $s_0$.
\end{proposition}
\begin{proof}
  Recall that the block positivity condition is equivalent to
  $(x\otimes y, q x\otimes y)\le 1/2$. Take $x=1/\sqrt{2}(e_1+e_2)=y$,
  then $(x\otimes y, P_{x_1} x\otimes y) = 1/2$. Thus $(x\otimes y,
  P_{x_2} x\otimes y)=0$. By the analogous argument we infer that
  $(e_1\otimes e_1,P_{x_2} e_1\otimes e_1)=0$ and $(e_2\otimes e_2,
  P_{x_2} e_2\otimes e_2)=0$. Consequently $x_2$ must belong to the
  subspace spanned by six basis vectors $e_1\otimes e_3, e_2\otimes
  e_3, e_3\otimes e_1, e_3\otimes e_2, e_3\otimes e_3$ and
  $1/\sqrt2(e_1\otimes e_2-e_2\otimes e_1)$. Due to assumed form of $x_2$ we have following possibilities (without normalization constant)
  \begin{align*}
  (i)\quad& e_1\otimes e_3\pm e_2\otimes e_3 & (ii)\quad& e_1\otimes e_3\pm e_3\otimes e_1 & (iii)\quad& e_1 \otimes e_3\pm e_3\otimes e_2 \\
  (iv)\quad& e_1\otimes e_3\pm e_3\otimes e_3 & (v)\quad& e_2\otimes e_3\pm e_3\otimes e_1 & (vi)\quad& e_2\otimes e_3\pm e_3\otimes e_2 \\
  (vii)\quad& e_2\otimes e_3 \pm e_3\otimes e_3 & (viii)\quad& e_3\otimes e_1\pm e_3\otimes e_2 & (ix)\quad& e_3\otimes e_1 \pm e_3\otimes e_3 \\
  (x)\quad& e_3\otimes e_2\pm e_3\otimes e_3 & (xi)\quad& e_1\otimes e_2 - e_2\otimes e_1
  \end{align*}
  Note that (i), (iv), (vii), (viii), (ix), (x) are Schmidt rank 1 vectors, thus it remains to consider
  \begin{enumerate}
  \item $e_1\otimes e_3\pm e_3\otimes e_1$,
  \item $e_1\otimes e_3\pm e_3\otimes e_2$,
  \item $e_2\otimes e_3\pm e_3\otimes e_1$,
  \item $e_2\otimes e_3\pm e_3\otimes e_2$,
  \item $e_1\otimes e_2 - e_2\otimes e_1$.
  \end{enumerate}
  Note that it suffices to consider only `+' case, as we can get minus
  by performing local unitary transformation $e_3\to-e_3$ on the first
  Hilbert space component and leave everything else unchanged.
  
  Let us examine the first case. We will use once more the block
  positivity condition. Take $x=1/\sqrt{2}(e_1+e_3)=y$. Then one
  calculates $(x\otimes y, P_{x_2} x\otimes y) = \frac12$
  but $(x\otimes y, P_{x_1} x\otimes y) = 1/8$, thus this violates block
  positivity and the first case is excluded. By the analogous
  argument we also exclude the fourth case.

  In the second case $x_3$ can be a linear combination of the remaining
  $e_2\otimes e_3, e_3\otimes e_1, e_3\otimes e_3$ and $1/\sqrt2 (e_1\otimes
  e_2 - e_2\otimes e_1)$ (as it must be orthogonal to $x_1$ and $x_2$). When
  we consider $x=1/\sqrt{2}(e_1+e_3), y=1/\sqrt{2}(e_2+e_3)$ we get $(x\otimes
  y, P_{x_2} x\otimes y)=\frac12$, what excludes $e_3\otimes e_3$ from the
  list (as in that case $(x\otimes y, P_{x_3} x\otimes y)>0$). Now consider
  $\Tr(\bbone\otimes P_{e_2})(P_{x_1} + P_{x_2}) = 1$. From Lemma
  \ref{thm:alpha-norm} it follows that $\alpha$-normalization demands that
  $\Tr(\bbone\otimes P_{e_2}) P_{x_3} = 0$ what excludes $1/\sqrt2
  (e_1\otimes e_2 - e_2\otimes e_1)$ from $x_3$. Thus, $x_3 =
  1/\sqrt{2}(e_2\otimes e_3\pm e_3\otimes e_1)$. In fact there must be a `-'
  sign, as for $x=1/\sqrt{3}(e_1+e_2+e_3)$ plus sign gives $(x\otimes x, q
  x\otimes x)=2/3$. In the third case $x_2$ and $x_3$ are swapped. 

  Finally in the last case we have that $\Tr (\bbone\otimes
  P_{e_1})(P_{x_1}+P_{x_2})=1$ and $\Tr (\bbone\otimes
  P_{e_2})(P_{x_1}+P_{x_2})=1$, thus $\Tr (\bbone\otimes P_{e_1})P_{x_3}$ and
  $\Tr (\bbone\otimes P_{e_2})P_{x_3}$ must equal zero (again we use
  $\alpha$-normalization and Lemma \ref{thm:alpha-norm}). This cannot be true
  if $x_3$ has Schmidt rank 2, so we arrive to contradiction.
\end{proof}

\begin{remark}
  Notice that in the Proposition above we indeed put restriction only
  on one vector, namely $x_2$. Form of $x_1$ only specifies the
  basis. 
\end{remark}

\begin{lemma}
\label{thm:rank-of-bad-and-good}
  Let $x=1/\sqrt2(e_1\otimes e_2-e_2\otimes e_1)$ and
  $y=1/\sqrt2(g_1\otimes e_3+e_3\otimes g_2)$, where
  $g_1,g_2\in\lspan\set{e_1,e_2}$ (all vectors are normalized). Then
  $z=c_1 x + c_2 y$ for any nonzero $c_1,c_2\in\complexes$ have Schmidt
  rank equal to $3$ unless $g_1=\pm g_2$. 
\end{lemma}
\begin{proof}
  Firstly let us rewrite
  \begin{align*}
    g_1 &= \sin\alpha\; e_1 + e^{i\phi}\cos \alpha\; e_2,\\
    g_2 &= \sin\beta\;  e_1 + e^{i\psi}\cos \beta\;  e_2,\\
    y   &= \frac1{\sqrt2}(g_1\otimes e_3 + e^{i\eta} e_3 \otimes g_2),\\
    z   &= c (\cos\gamma\; x + e^{i\chi} \sin\gamma\; y).
  \end{align*}
  This is exactly equivalent to the statement of the theorem, but now
  parameters $\alpha,\beta,\gamma,\phi,\psi,\chi,\eta$ are real, and
  only $c$ is complex. Now recall that the Schmidt rank of the vector $z$
  is equal to the rank of the matrix formed by coefficients $z_{ij} =
  (e_i\otimes e_j, z)$. Thus if $z$ have Schmidt rank less than 3,
  then $\det (z_{ij})=0$. By the explicit calculation we have
  \begin{equation*}
    \det (z_{ij}) = -\frac1{2\sqrt2} 
    c^3 \sin ^2\gamma \cos \gamma e^{i (\eta +2 \chi )} \left(e^{i \psi }
      \sin \alpha \cos\beta-e^{i \phi } \cos\alpha\sin\beta\right)
  \end{equation*}
  This is equal to zero only when $(e^{i \psi }
  \sin \alpha \cos\beta-e^{i \phi } \cos\alpha\sin\beta) = 0$,
  i.e. either $\sin(\alpha-\beta)=0$ and $\phi-\psi=0$ or
  $\sin(\alpha+\beta)=0$ and $\phi-\psi=\pi$.  In the first case
  $g_1=g_2$ and in the second $g_1=-g_2$. 
\end{proof}

\begin{lemma}
  \label{thm:not-twisted}
  Let $\xi_1,\xi_2,\xi_3$ are three orthonormal vectors in $\hilbert
  H\otimes \hilbert H$ with Schmidt rank equal 2. Let us put
  $\xi_1=1/\sqrt2 (e_1\otimes e_1+e_2\otimes e_2)$. If one of
  remaining vectors is of the form 
  \begin{equation*}
    c_1 (e_1\otimes e_2-e_2\otimes
    e_1) + c_2 e_1\otimes e_3 + c_3 e_2\otimes e_3 + 
    c_4 e_3\otimes e_2 + c_5 e_3\otimes e_1,\quad \text{where }c_i\in\complexes
  \end{equation*}
  then the symmetry $s=\bbone - 2 q$, where $q=\sum_i P_{\xi_i}$ is not in
  $\mathfrak D$ unless $c_1=0$.
\end{lemma}
\begin{proof}
  Without loss of generality assume that $\xi_2$ is of
  the form above. Then it can be written in the same form as $z$ in
  the proof of the Lemma~\ref{thm:rank-of-bad-and-good}, precisely
  \begin{align*}
    x &= 1/\sqrt2 (e_1\otimes e_2-e_2\otimes e_1),\\
    g_1 &= \sin\alpha\; e_1 + e^{i\phi}\cos \alpha\; e_2,\\
    g_2 &= \sin\beta\;  e_1 + e^{i\psi}\cos \beta\;  e_2,\\
    y   &= \frac1{\sqrt2}(g_1\otimes e_3 + e^{i\eta} e_3 \otimes g_2),\\
    \xi_2   &= \cos\gamma\; x + e^{i\chi} \sin\gamma\; y.
  \end{align*}
  But then we easily see that $\Tr (\bbone\otimes P_{e_3})P_{\xi_2}=1/2
  \sin^2\gamma$. As $\Tr (\bbone\otimes P_{e_3}) P_{\xi_1}=0$, by the
  $\alpha$-normalization and Lemma \ref{thm:alpha-norm} we need $\Tr
  (\bbone\otimes P_{e_3})P_{\xi_3} = 1-\frac12 \sin^2\gamma$. But
  $1-\frac12\sin^2 \gamma > \frac12$ for $\gamma\neq\pi/2$ or $3/2\pi$. Then,
  by the Lemma~\ref{thm:schmidtcoeff} one of the Schmidt coefficients of
  $\xi_3$ would have to be greater than $1/\sqrt2$, so the thesis is proved
  (see Corollary \ref{thm:schmidt-rank2}). On the other hand if $\gamma=\pi/2$
  or $\gamma=3/2\pi$ then, turning back to notation from the statement of the
  lemma, we get $c_1=0$.
\end{proof}

\begin{lemma}
  \label{thm:not-e3-e3}
  Let $\xi_1,\xi_2,\xi_3$ are three orthonormal vectors in $\hilbert
  H\otimes \hilbert H$ with Schmidt rank equal 2. Let us put
  $\xi_1=1/\sqrt2 (e_1\otimes e_1+e_2\otimes e_2)$. If one of
  remaining vectors is of the form 
  \begin{equation*}
    c_1 e_1\otimes e_3 + c_2 e_2\otimes e_3 + 
    c_3 e_3\otimes e_2 + c_4 e_3\otimes e_1 + c_5 e_3\otimes e_3,\quad\text{where } c_i\in\complexes
  \end{equation*}
  then the symmetry
  $s=\bbone - 2 q$, where $q=\sum_i P_{\xi_i}$ is not in $\mathfrak D$
  unless $c_5=0$.
\end{lemma}
\begin{proof}
  Without loss of generality let us assume that the $\xi_2$ is of the form
  above. The $\alpha$-normalization demands that $\Tr (\bbone\otimes P_{e_3})
  q = 1$ but $\Tr (\bbone\otimes P_{e_3}) P_{\xi_1}=0$, thus we infer that
  $\Tr (\bbone\otimes P_{e_3}) P_{\xi_2} = 1/2 = \Tr (\bbone\otimes P_{e_3})
  P_{\xi_3}$ (note that the trace cannot be greater than $1/2$ as the maximal
  Schmidt coefficient is equal to $1/\sqrt2$, cf. Lemma
  \ref{thm:schmidtcoeff}). Then
  \begin{equation*}
    \Tr (\bbone\otimes P_{e_3}) P_{\xi_2} = \abs{c_1}^2 + \abs{c_2}^2 +
    \abs{c_5}^2 = \frac{1}{2}.
  \end{equation*}
  So $\abs{c_3}^2+\abs{c_4}^2=1/2$ due to normalization. Now we are
  going to show that $c_5$ must equal to zero. Consider the family
  of vectors of the form:
  \begin{equation*}
    u_\phi = c_4 e_1+c_3 e_2 + e^{i \phi}/\sqrt{2} e_3.
  \end{equation*}
  Now we directly calculate
  \begin{equation*}
    \Tr (\bbone\otimes P_{u_\phi}) P_{\xi_2} = \frac{1}{2} +
    \frac{1}{\sqrt2} ( c_5\conj e^{i\phi} + c_5 e^{-i\phi})
  \end{equation*}
  As  $\Tr (\bbone\otimes P_{u_\phi}) P_{\xi_2}$ must be less or equal
  to $1/2$, then
  \begin{equation*}
    c_5\conj e^{i\phi} + c_5 e^{-i\phi} \le 0.
  \end{equation*}
  If we take $\phi=0$, we get that $\Re c_5 \le 0$. For $\phi=\pi$ we
  get $\Re c_5 \ge 0$, $\phi=\pi/2$ implies $\Im c_5 \ge 0$ and
  $\phi=3/2 \pi$ gives $\Im c_5 \le 0$. Thus $c_5 = 0$.
\end{proof}

\begin{lemma}
  \label{thm:not-both}
  Let $\xi_1,\xi_2,\xi_3$ are three orthonormal vectors in $\hilbert
  H\otimes \hilbert H$ with Schmidt rank equal 2. Let us put
  $\xi_1=1/\sqrt2 (e_1\otimes e_1+e_2\otimes e_2)$. If one of
  remaining vectors is of the form 
  \begin{equation*}
    c_1\frac1{\sqrt 2}(e_1\otimes e_2 - e_2\otimes e_1) +
    c_2 e_1\otimes e_3 + c_3 e_2\otimes e_3 + 
    c_4 e_3\otimes e_2 + c_5 e_3\otimes e_1 + c_6 e_3\otimes e_3,
  \end{equation*}
  then the symmetry
  $s=\bbone - 2 q$, where $q=\sum_i P_{\xi_i}$ is not in $\mathfrak D$
  if both $c_1$ and $c_6$ are not equal to 0.
\end{lemma}
\begin{proof}
  We will prove the statement by the contradiction, thus assume that
  $c_1\neq 0$, $c_6\neq 0$ and $s$ is in $\mathfrak D$.  Firstly, let
  us assume that $\xi_2$ is of the claimed form, and we get rid off
  irrelevant overall phase factor assuming that $c_6$ is real. We know
  from Lemma \ref{thm:rank2fullent} that $\xi_2$ must be Schmidt rank
  2 with equal Schmidt coefficients. Then, it is well known that the
  coefficient matrix $(a_{ij})$, where $a_{ij}\defeq (e_i\otimes e_j,
  \xi_2)$ must have zero determinant (otherwise it would be full rank
  and this would mean that $\xi_2$ must by Schmidt rank 3).
  By the explicit calculation one finds that
  \begin{equation*}
    \det (a_{ij}) = \frac{1}{2} c_1 \left(c_6 c_1 + \sqrt{2} (c_3 c_5
      - c_2 c_4)\right).
  \end{equation*}
  Let us denote by $\delta= (c_3 c_5 - c_2 c_4)$. As we assumed that
  $c_1\neq 0$, this means that
  \begin{equation*}
    c_1 c_6 + \sqrt{2} \delta = 0.
  \end{equation*}
  We can multiply it by $\cconj c_1$ to get $\abs{c_1}^2 c_6 +
  \sqrt{2} \delta\cconj c_1 = 0$. Now adding this equation and its conjugate
  together we can express the necessary condition for
  $\xi_2$ to be Schmidt rank 2 as (we remind that $c_6$ is chosen to be real)
  \begin{equation}
    \label{eq:det-zero}
    \sqrt{2} (\delta \cconj c_1 + \cconj \delta c_1) = - 2 \abs{c_1}^2 c_6.
  \end{equation}
  
  Now note that as we demand that $\xi_2$ is Schmidt rank 2 with equal
  Schmidt coefficients we can assume that $\xi_2 = 1/\sqrt2
  (f_1\otimes g_1 + f_2\otimes g_2)$ for appropriate vectors $f_i$ and
  $g_i$.  Now take the projector $P_{\xi_2}$ onto the vector $\xi_2$
  and calculate partial trace $\omega = \Tr_2 P_{\xi_2}$ with respect
  to the second Hilbert space. Partial trace does not depend on basis
  and we see that eigenvalues of $\omega$ are equal to squares of
  Schmidt coefficients of $\xi_2$, thus must be equal to $1/2$. On the
  other hand we know that the eigenvalues of $\omega$ are
  roots of characteristic polynomial, which in case of $3\times 3$
  matrix can be written in the form
  \begin{equation*}
    \det(\lambda\bbone - \omega) = \lambda^3 + a \lambda^2 + b \lambda
    + d = 0.
  \end{equation*}
  The $d$ is simply equal to $\det(\omega)$ and must equal to zero,
  otherwise $\xi_2$ would be Schmidt rank 3 vector. Thus
  \begin{equation*}
    \lambda (\lambda^2 + a \lambda + b) = 0.
  \end{equation*}
  
  Let us for while forget about the assumptions that two non-zero roots must be
  equal. As always here the non-zero solutions $\lambda_1$ and $\lambda_2$ are
  the squares of Schmidt coefficients, which are positive and their squares
  sum up to 1 (due to normalization), we can assume that
  \begin{equation*}
    \lambda_1 = \sin^2 \theta,\quad \lambda_2 = \cos^2\theta,\quad
    0<\theta<\frac\pi2.
  \end{equation*}
  Then we immediately see that $\lambda_1\lambda_2\leq1/4$ and is
  equal $1/4$ if and only if $\sin\theta = 1/\sqrt2 = \cos\theta$, thus in
  case of equal Schmidt coefficients. We then use Vieta's formula to
  express necessary condition for a $\lambda_1$ and $\lambda_2$ to be equal
  \begin{equation*}
    \lambda_1\lambda_2 = b = \frac14
  \end{equation*}
  By the explicit calculation one find that $b$ can be expressed as
  \begin{equation*}
    \begin{split}
      b = \frac{1}{4} \bigg(2\sqrt{2}c_6 \left( (c_3 c_5-c_2
            c_4) \cconj c_1 + c_1 (\cconj c_3 \cconj c_5 - \cconj
            c_2\cconj c_4)\right) + 4
          \left(\abs{c_2}^2+\abs{c_3}^2\right)
          \left(\abs{c_4}^2+\abs{c_5}^2\right)\\ +2 |c_1|^2
        \left(2
          c_6^2+|c_2|^2+|c_3|^2+|c_4|^2+|c_5|^2\right)+|c_1|^4\bigg).
    \end{split}
  \end{equation*}
  To simplify this expression we will use the fact that $\alpha$-normalization
  of $s$ demands that $\Tr (\bbone\otimes P_{e_3})P_{\xi_2} = 1/2$ (cf. proof
  of Lemma \ref{thm:not-e3-e3}). By explicit calculation this means that
  \begin{equation*}
    \abs{c_2}^2+\abs{c_3}^2+c_6^2 = \frac12,
  \end{equation*}
  and this combined with normalization of vector $\xi_2$ yields
  \begin{equation*}
    \abs{c_1}^2+\abs{c_4}^2+\abs{c_5}^2 = \frac12.
  \end{equation*}
  We will use these equalities to eliminate $\abs{c_2}^2, \abs{c_3}^2,
  \abs{c_4}^2$ and $\abs{c_5}^2$ from the $b$. We will also substitute
  $\delta$ where it applies. We get
  \begin{equation*}
    b = \frac14 \left( 2\sqrt2 c_6(\delta\cconj c_1 + \cconj
      \delta c_1) + 4 \left(\textstyle\frac12-c_6^2\right)\left(\textstyle
\frac12-\abs{c_1}
^2\right) +
      2 \abs{c_1}^2 (c_6^2 - \abs{c_1}^2 + 1) + \abs{c_1}^4\right).
  \end{equation*}
  Now we substitute eq.~\eqref{eq:det-zero} and simplify
  expression to get
  \begin{equation*}
    b = 1/4 \left(2\abs{c_1}^2c_6^2 - 2 c_6^2 - \abs{c_1}^4+1\right)
    = 1/4.
  \end{equation*}
  This is satisfied when
  \begin{equation*}
    2\abs{c_1}^2c_6^2 - 2 c_6^2 - \abs{c_1}^4 = 0
  \end{equation*}
  To simplify notation denote $\abs{c_1}^2 = \alpha$ and
  $c_6^2=\beta$. Now one can immediately see that the equation
  \begin{equation*}
    2 \alpha\beta - 2\beta-2\alpha^2 = 0
  \end{equation*}
  does not have solutions for
  $0<\alpha,\beta<1$. Thus we arrived to the contradiction.
\end{proof}

\begin{corollary}
  \label{thm:form-of-s}
  If a symmetry $s=\bbone - 2 q$, where $q=\sum_i P_{x_i}$ is in $\setd$ then
  \begin{align*}
    x_1 &= \frac{1}{\sqrt{2}} (e_1\otimes e_1 + e_2\otimes e_2),\\
    x_2 &= \frac{1}{\sqrt{2}} (g_1\otimes e_3 + e_3\otimes h_2),\\
    x_3 &= \frac{1}{\sqrt{2}} (k_1\otimes e_3 + e_3\otimes l_2),
  \end{align*}
\end{corollary}
\begin{proof}
  We know that $x_2$ and $x_3$ must be a linear combination of the
  form
  \begin{equation*}
    c_1\frac1{\sqrt 2}(e_1\otimes e_2 - e_2\otimes e_1) +
    c_2 e_1\otimes e_3 + c_3 e_2\otimes e_3 + 
    c_4 e_3\otimes e_2 + c_5 e_3\otimes e_1 + c_6 e_3\otimes e_3,
  \end{equation*}
  but Lemma~\ref{thm:not-both}, \ref{thm:not-e3-e3} and
  \ref{thm:not-twisted} imply together that $c_1=0$ and $c_6=0$. Then
  we get the desired form.
\end{proof}

\begin{lemma}
  \label{thm:perp}
  Let $s = \bbone - 2 q\,\in\mathfrak D$ be the symmetry where
  $q=\sum_i P_{x_i}$ and
  \begin{align*}
    x_1 &= \frac{1}{\sqrt{2}} (e_1\otimes e_1 + e_2\otimes e_2),\\
    x_2 &= \frac{1}{\sqrt{2}} (g_1\otimes e_3 + e_3\otimes h_2),\\
    x_3 &= \frac{1}{\sqrt{2}} (k_1\otimes e_3 + e_3\otimes l_2),
  \end{align*}
  then $h_2\perp l_2$ and $g_1\perp k_1$.
\end{lemma}
\begin{proof}
  Let us calculate $\Tr (\bbone\otimes P_{h_2})P_{x_i}$. For $P_{x_1}$
  we get $1/2$, as $h_2$ belongs to the span of $e_1$ and $e_2$ (it is
  shown in the previous proof). Obviously for $P_{x_2}$ this also
  equals $1/2$, so for $P_{x_3}$ it must equal $0$. We thus calculate
  $\Tr (\bbone\otimes P_{h_2})P_{x_3}$ (once more we use the
  fact, that non-diagonal terms in the first tensor product factor
  will vanish)
  \begin{equation*}
    \begin{split}
      \Tr (\bbone\otimes P_{h_2})P_{x_3} = 
      \frac12 \left(\braket{h_2|e_3}\braket{e_3|h_2} +
      \braket{h_2|l_2}\braket{l_2|h_2}\right) = 
      \frac12 \abs{(h_2,l_2)}^2
    \end{split}
  \end{equation*}
  so the first claim
  follows. Then using the fact that $(x_2,x_3)=0$, we get
  \begin{equation*}
    (x_2,x_3)=\frac12(g_1\otimes e_3 + e_3\otimes h_2, k_1\otimes e_3 +
    e_3\otimes l_2) 
    = \frac12 \left( (g_1,k_1)+(h_2,l_2)\right) = \frac12 (g_1,k_1),
  \end{equation*}
  and the second claim follows.
\end{proof}
\begin{theorem}
  For $n=2,3$ any symmetry in $\mathfrak D$ is locally unitarily
  equivalent to Choi matrix corresponding to the transposition map. 
\end{theorem}
\begin{proof}
  For $n=2$ the result is already known (see
  \cite{Majewski:2010fk}). For $n=3$, take a symmetry $s\in\mathfrak
  D$. Denote $s = \bbone - 2 q$ where $q=\sum_i P_{y_i}$. Then from
  the Corollary~\ref{thm:form-of-s} we know that 
  \begin{align*}
    y_1 &= \frac{1}{\sqrt{2}} (e_1\otimes e_1 + e_2\otimes e_2),\\
    y_2 &= \frac{1}{\sqrt{2}} (g_1\otimes e_3 + e_3\otimes h_2),\\
    y_3 &= \frac{1}{\sqrt{2}} (k_1\otimes e_3 + e_3\otimes l_2).
  \end{align*}
  According to the Lemma~\ref{thm:perp} $h_2\perp l_2$ and
  $g_1\perp k_1$. Moreover we know that $g_1\perp e_3$, $k_1\perp
  e_3$, $h_2\perp e_3$ and $l_2\perp e_3$. Thus $\set{g_1,k_1,e_3}$
  and $\set{h_2,l_2,e_3}$ are two sets of mutually orthogonal vectors that
  we can consider as a bases in corresponding Hilbert spaces. This
  allows us to define two unitary operators:
  \begin{align*}
    U g_1 &= e_1,& U k_1 &= -e_3,& U e_3 &= e_2,\\
    V h_2 &= e_2,& V l_2 &= e_3,& V e_3 &= e_1.
  \end{align*}
  Then
  \begin{align*}
    U\otimes V y_1 &= \frac{1}{\sqrt{2}}( U e_1\otimes V e_1+U e_2\otimes
    V e_2),\\
    U\otimes V y_2 &= \frac{1}{\sqrt{2}}( e_1\otimes e_1 + e_2\otimes
    e_2),\\
    U\otimes V y_3 &= \frac{1}{\sqrt{2}}( -e_3\otimes e_1 + e_2\otimes
    e_3).
  \end{align*}
  These three vectors satisfy assumptions of the
  Proposition~\ref{thm:twofixed}, so $s$ is locally unitarily
  equivalent to $s_0$. But $s_0$ is locally unitarily equivalent to
  the Choi matrix of transposition map and the claim follows.
\end{proof}

\begin{remark}
  Consider case $n=3$. It is clear that any antisomorphism is
  represented by a symmetry. Above theorem establishes the converse:
  any symmetry corresponds to the antisomorphism. Consequently any
  isomorphism in $\pmaps(\malg 3)$ is represented by a Choi matrix of
  the form $3 P_x$, for some maximally entangled vector $x$.
\end{remark}

\subsection{Symmetries as exposed points of $\setd$}

Our goal is to show that the Choi matrix corresponding to transposition map in
$\malg 3$ is an exposed (so also an extreme) point of $\setd$. We start with a
lemma.

\begin{lemma}
  \label{thm:wex-lemma}
  Let $w = \sum_{i,j=1}^{n} E_{ij}\otimes E_{ji}$ and 
  $\sigma \in \setd$ such that $\Tr w \sigma = n^2$. Then
  \begin{enumerate}[(i)]
  \item $(e_j, \sigma_{ij} e_i) = 1$,
  \item $\sum_i \sigma_{ii} = \bbone$,
  \item $\sigma_{ii}\ge 0$, and $(e_i, \sigma_{jj} e_i) = 0$ for $i\neq j$, 
  \item $(e_i, \sigma_{ij} e_i) = 0$ for $i\neq j$,
  \end{enumerate}
  where we adopted notation $\sigma = \sum_{ij} E_{ij}\otimes \sigma_{ij}.$
\end{lemma}
\begin{proof}
  For (i) we firstly note that
  \[
    \Tr w \sigma = \sum_{i,j=1}^n (e_i\otimes e_j, \sigma e_j\otimes e_i),
  \]
  but on the other hand due to $\alpha(\sigma) = 1$, i.e. 
  \[\sup\set{\abs{\Tr
  \sigma a}\setdef a\in\malg 3, \pi(a) = 1} = 1,
  \]
  one has
  \[
    \abs{\Tr \sigma E_{ji} \otimes E_{ij}} = 
    \abs{(e_i\otimes e_j, \sigma e_j\otimes e_i)} \leq 1.
  \]
  These and the assumption that $\Tr w \sigma = n^2$ implies (i).

  Property (ii) follows immediately:
  \[
    \sum_i \sigma_{ii} = \sum_i \phi(E_{ii}) = \phi(\bbone) = \bbone,
  \]
  where $\phi$ is a positive unital normalized map corresponding to $\sigma$. 

  To show (iii) we need to apply block-positivity condition $\sigma\gebp 0$. In
  particular
  \[
    0\le (e_m\otimes y, \sigma e_m\otimes y) = \sum_{ij} (e_m\otimes y, E_{ij}\otimes \sigma_{ij} e_m\otimes y) = (y,\sigma_{mm} y).
  \]
  Now, due to (ii) $(e_k, \sum_{m} \sigma_{mm} e_k) = 1$, due to (i) $(e_k,  
  \sigma_{kk} e_k) = 1$ and due to last inequality $(e_k,\sigma_{mm} e_k)\ge  
  0$. Thus we obtained desired result. 

  Finally to show (iv) we proceed as in the proof of Prop. \ref{thm:form-of-rho}. We take $x=\epsilon e_i + \lambda e_j$ and $y=e_i$ with $i\neq j$ and $\epsilon>0,\lambda\in\reals$ (so the vectors are not necessarily normalized). Then block positivity gives us
  \begin{equation*}
    \begin{split}
      0&\le (x\otimes e_i, \sum_{kl} E_{kl} \otimes \sigma_{kl} x\otimes e_i)\\
      &=\epsilon^2 (e_i,\sigma_{ii} e_i) + \epsilon \lambda (e_i,\sigma_{ji} e_i)
      + \epsilon \lambda (e_i,\sigma_{ji} e_i) + \lambda^2 (e_i,\sigma_{jj}e_i)
    \end{split}
  \end{equation*}
  Due to our assumptions and results already obtained this means that
  \begin{equation*}
    \lambda (e_i,(\sigma_{ij}+\sigma_{ji}) e_i)\ge -\epsilon.
  \end{equation*}
  Repeating this for $x=\epsilon e_i - \lambda e_j, x = \epsilon e_i + i \lambda e_j, x = \epsilon e_i - i \lambda e_j$, we conclude that $(e_i,\sigma_{ij} e_i) = 0$ for $i\neq j$.
\end{proof}

\begin{theorem}
	The Choi matrix $w$ 
	\begin{equation}
		w = \sum_{ij} E_{ij}\otimes E_{ji}
	\end{equation}
	corresponding to transposition map in $\malg 3$ is an exposed point of $\setd$.
\end{theorem}
\begin{proof}
  We will show that the value of functional $\omega(\sigma) = \Tr w \sigma$ is strictly less than $n^2$ for $\sigma\in\setd$ unless $\sigma = w$.

  Because $\Tr w \sigma = \sum_{i,j=}^n (e_i\otimes e_j, \sigma e_j\otimes e_i)$ and due to $\alpha$-normalization of $\sigma$ we have that $\Tr w \sigma\leq n^2$. It is clear that $\Tr w w = n^2$. Let us take arbitrary $\sigma\in\setd$ such that $\Tr w \sigma = n^2$. From the previous lemma we know that $\sigma_{ii}\ge 0$ and 
  \begin{equation*}
    (e_1,\sigma_{11} e_1) = 1,\quad (e_2,\sigma_{11} e_2) = 0,\quad (e_3,\sigma_{11} e_3) = 0,
  \end{equation*}
  from which we infer that $e_2,e_3\in \ker \sigma_{11}$, so $\sigma_{11} = \ket{e_1}\bra{e_1}.$ Analogously we show that $\sigma_{22} = \ket{e_2}\bra{e_2}$ and $\sigma_{33} = \ket{e_3}\bra{e_3}.$ 

  Now let us consider $\sigma_{12}$. From the Lemma~\ref{thm:wex-lemma} we immediately get that 
  \begin{equation}
    \label{eq:wex-proof-1}
    (e_2, \sigma_{12} e_1) = 1
  \end{equation}
  and
  \begin{equation}
    \label{eq:wex-proof-2}
    (e_1, \sigma_{12} e_1) = 0.
  \end{equation}
  Due to the fact that $\sigma$ is hermitian, $(e_2, \sigma_{12} e_2) = (\sigma_{21} e_2, e_2) = \cconj{(e_2, \sigma_{21} e_2)} = 0$, so
  \begin{equation*}
    \label{eq:wex-proof-3}
    (e_2, \sigma_{12} e_2) = 0
  \end{equation*}

  Now we proceed as in the proof of Prop.~\ref{thm:form-of-rho}. Precisely, take $x_\pm = \epsilon e_1 \pm \lambda e_2,$ with $\epsilon>0, \lambda\in\reals$ and $y_\pm = e_1\pm e_3$. Then
  \begin{equation*}
     0\le (x_+\otimes y_+, \sigma x_+\otimes y_+) + (x_+\otimes y_-, \sigma x_+\otimes y_-) = 2 \epsilon^2 + 4 \epsilon \lambda \Re (e_3, \sigma_{12} e_3)
   \end{equation*} 
   and
   \begin{equation*}
     0\le (x_-\otimes y_+, \sigma x_-\otimes y_+) + (x_-\otimes y_-, \sigma x_-\otimes y_-) = 2 \epsilon^2 - 4 \epsilon \lambda \Re (e_3, \sigma_{12} e_3),     
   \end{equation*}
 so $-\epsilon\leq 2 \lambda \Re (e_3, \sigma_{12} e_3) \leq \epsilon$. Due to
 arbitrariness of $\lambda$, $\Re (e_3, \sigma_{12} e_3) = 0$. Analogous
 calculations for $u_\pm = \epsilon e_1 \pm i \lambda e_2$ instead of $x_\pm$
 yield that $\Im (e_3, \sigma_{12} e_3) = 0$, so
   \begin{equation*}
     (e_3, \sigma_{12} e_3) = 0.
   \end{equation*}

   Using these results we see that
   \begin{align*}
     0&\le (x_+\otimes y_+, \sigma x_+\otimes y_+) = \epsilon^2 + 2 \epsilon \lambda \Re \left( (e_1,\sigma_{12} e_3) + (e_3, \sigma_{12} e_1) \right) \\
     0&\le (x_+\otimes y_-, \sigma x_+\otimes y_-) = \epsilon^2 - 2 \epsilon \lambda \Re \left( (e_1,\sigma_{12} e_3) + (e_3, \sigma_{12} e_1) \right) \\
     0&\le (u_+\otimes y_+, \sigma u_+\otimes y_+) = \epsilon^2 - 2 \epsilon \lambda \Im \left( (e_1,\sigma_{12} e_3) + (e_3, \sigma_{12} e_1) \right) \\
     0&\le (u_+\otimes y_-, \sigma u_+\otimes y_-) = \epsilon^2 + 2 \epsilon \lambda \Im \left( (e_1,\sigma_{12} e_3) + (e_3, \sigma_{12} e_1) \right) \\
   \end{align*}
   so $(e_1,\sigma_{12} e_3) + (e_3, \sigma_{12} e_1) = 0$. Analogous results for $v_\pm = e_1\pm i e_3$ yield that $(e_1,\sigma_{12} e_3) - (e_3, \sigma_{12} e_1)=0$, so
   \begin{equation*}
     \label{eq:wex-proof-4}
     (e_1,\sigma_{12} e_3) = 0,\quad (e_3, \sigma_{12} e_1) = 0.
   \end{equation*}

  Repeating the same arguments for $y_\pm = e_2 \pm e_3, v_\pm = e_2\pm i e_3$ we get that
   \begin{equation*}
     \label{eq:wex-proof-5}
     (e_2, \sigma_{12} e_3) = 0,\quad (e_3, \sigma_{12} e_2) = 0.
   \end{equation*}

  It remains to show that $(e_1, \sigma_{12} e_2) = 0$. Firstly we take $\epsilon = 1,\lambda = 1$, $y_\pm = e_1\pm e_2$ and $v_\pm = e_1\pm i e_2$ and see that
  \begin{align*}
    0&\le (u_-\otimes v_+, \sigma u_-\otimes v_+) = 2\Re (e_1, \sigma_{12} e_2),\\
    0&\le (x_+\otimes y_-, \sigma x_+\otimes y_-) = -2\Re (e_1, \sigma_{12} e_2),
  \end{align*}
  so $\Re (e_1, \sigma_{12} e_2) = 0$. Now for $z_\pm = \epsilon e_1 \pm (1\pm \epsilon i) e_2$ we calculate (using previous results) that
	\begin{align*}
		0 &\le (z_+\otimes v_-, \sigma z_+\otimes v_-) = 
			1-2i \epsilon (e_1, \sigma_{12} e_2),\\
		0 &\le (z_-\otimes v_-, \sigma z_-\otimes v_-) = 
			1+2i \epsilon (e_1, \sigma_{12} e_2),\\
	\end{align*}
	so that for every $\epsilon>0$
	\begin{equation*}
		-\frac1{2 \epsilon} \le \Im (e_1, \sigma_{12} e_2) \le \frac1{2 \epsilon}
	\end{equation*}
  We conclude that $\Im (e_1, \sigma_{12} e_2) = 0$. Gathering all those results together we see that 
  \begin{equation*}
    \sigma_{12} = \ket{e_2}\bra{e_1}.
  \end{equation*}
  Using the same methods we show that
  \begin{equation*}
    \sigma_{13} = \ket{e_3}\bra{e_1},\quad \sigma_{23} = \ket{e_3}\bra{e_2}.
  \end{equation*}
  Thus if for any $\sigma\in\setd$, $\Tr w \sigma = n^2,$ then $\sigma=w$, otherwise $\Tr w \sigma < n^2$, so $w$ is an exposed point of $\setd$.
\end{proof}

Combining this result with previous section we see that.

\begin{corollary}
  Any symmetry $s\in\setd$ is an exposed point of $\setd$. Also any Choi matrix
  of the form $3 P_x$, where $x$ maximally entangled vector, is an exposed
  point of $\setd$.
\end{corollary}

\begin{remark}
  This corollary immediately follows from results in \cite{marciniak2011rank}
  and repeats the result already given in \cite{Majewski:2010fk} (which was
  obtained via convex analysis). Also the criterion given in
  \cite{chruscinski2012exposed} shows that transposition map is an exposed
  map. Moreover the proof of mentioned criterion allows us to construct other
  functionals 'supporting' exposedness of $w$, so such functionals are far
  from being unique. Despite those two overlaps we decided to presented the
  longer proof to make it more consistent with
  Section~\ref{sec:extr-posit-maps} and emphasizes some similarities between
  $n=2$ and $n=3$ cases.
\end{remark}

\subsection{Partial symmetries} 
\label{sec:part-sym}

It is easy to see that for $n=2$ there can be no partial symmetries belonging
to $\setd$. For $n=3$ the situation is different. The unitality condition $\Tr
\rho = 3$ and decomposition $\rho = p - q$ imply that $e=p+q$ must be of rank
5 or 7. Moreover, it is known that for $n=4$ maps corresponding to partial
symmetries can be exposed and indecomposable, see \cite{MR813392} and
\cite{sarbicki2012class}. This advocates the importance of examination of
partial symmetries in $n=3$ case. 

The easiest example of block positive symmetry can be obtained by perturbation
of swapping operator in $n=2$ embedded in $n=3$.

\begin{example}
  Let
  \begin{equation*}
    w = \sum_{i,j=1}^2 E_{ij}\otimes E_{ji},\qquad\text{and } 
    x = \frac1{\sqrt2} (e_1 +e_2)\otimes e_3
  \end{equation*}
  then
  \begin{equation*}
    s_0 = w + P_x
  \end{equation*}
  is an e-symmetry. The rank of $s_0^2$ is equal to 5. This map is coCP, as
  partial transpose of $s_0$ is a positive matrix. 
\end{example}

One observe that that the image of map corresponding to $s_0$ is $4$
dimensional subspace of $\malg 3$. Thus in fact it is a map of the form $\malg
3\to\malg2\hookrightarrow\malg3$. It is known that any map $\malg3\to\malg2$ is
decomposable. It is natural to ask if there are other e-symmetries such that
corresponding maps are not of this form. Affirmative answer is given by the
following examples.

\begin{example}
  Let $w$ be as in previous example and $x=e_3\otimes e_3$. Let us put
  \begin{equation*}
    s = w + P_x 
  \end{equation*}
  Then $s^2$ is a partial symmetry with $s^2$ of rank 5. It is block positive
  because partial transpose of $s$ is positive. Thus the above map is coCP.
  The image of this map is a 5 dimensional subspace of $\malg3$. 
\end{example}

Despite extensive study we did not find any example of partial symmetry $s$ in $\setd$ for $n=3$ with rank of $s^2$ equal to 7 nor we didn't found any partial symmetry corresponding to a non-decomposable map. This led us to following conjecture.

\begin{conjecture}
  For $n=3$, if $s\in\setd$ is a partial symmetry then (i) rank of $s^2$ is equal to 5, (ii) $s$ corresponds to decomposable positive map.
\end{conjecture}


\section{Final remarks} 
\label{sec:fr}

Up to now we have studied block positive symmetries as well as Choi matrices of the form $n P_x$, where $x$ is maximally entangled vector. To complete the picture let us focus for a while on Choi matrices of the form $p\otimes \bbone$.

\begin{proposition}
  \label{thm:p1-ext}

  Choi matrices of the form $p\otimes \bbone$, where $p$ is rank 1 projector,
  are extreme points of $\setd$ for any $n$.
\end{proposition}
\begin{proof}
  Suppose that $p\otimes \bbone = \lambda \sigma_1 + (1 - \lambda) \sigma_2$
  and $p=\ket{f}\bra{f}$. Then for any vector $g$ we have 
  \begin{equation}
    \label{eq:conv-comb}
    (f\otimes g, p\otimes\bbone f\otimes g) = 1 = \lambda (f\otimes g, \sigma_1 f\otimes g)
    + (1 - \lambda)(f\otimes g, \sigma_2 f\otimes g).
  \end{equation}
  Due to $\alpha$-normalization of $\sigma_i$ we have $\Tr (\bbone\otimes P_g)
  \sigma_i = 1$. So
  \begin{equation*}
    \Tr P_f \otimes P_g \sigma_i + \Tr (\bbone - P_f)\otimes P_g \sigma_i = 1
  \end{equation*}
  Due to block positivity of $\sigma_i$ both terms must be greater than 0, so
  \begin{equation*}
    (f\otimes g, \sigma_i f\otimes g) \le 1,
  \end{equation*}
  but due to \eqref{eq:conv-comb} we need to have equality. Consequently for any $f'$
  orthogonal to $f$ and any $g$ we have that $(f'\otimes g, \sigma_i f'\otimes
  g)=0,$ so both $\sigma_i$ must equal $p\otimes\bbone$, so $p\otimes\bbone$
  is extremal.
\end{proof}

Consequently we conclude that the set
\begin{equation*}
  \tilde\setd = \set{\text{symmetries}, n P_x, p\otimes \bbone}
\end{equation*}
naturally arise as a subset of extremal Choi matrices. In fact we have shown
that symmetries, thus also Choi matrices of the form $n P_x$ are exposed
points of $\setd$ for $n=2,3$. We saw that the simple set $\tilde\setd$ is
enough to describe all regular extreme points of $\setd$ for $n=2$. We also
indicated how much deficient $\tilde\setd$ is for $n=3$ due to appearance of
\begin{enumerate}
\item partial symmetries (although their independence of $\tilde\setd$ for $n=3$ does not seem to be trivial and need further investigation),
\item non-decomposable maps, with Choi map as a standard example,
\item various concepts of extremality even when restricted to diagonal subalgebra.
\end{enumerate}


\section{Acknowledgments}

A partial support of the grant number N N202 208238 of Polish Ministry of
Science and Higher Education as well as of the grant of the Poland-South
Africa Cooperation Agreement is gratefully acknowledged. The contribution of
TT was supported within the International PhD Project "Physics of future
quantum-based information technologies", grant MPD/2009-3/4 from Foundation
for Polish Science. Authors are grateful to Marcin Marciniak and Paweł
Barbarski for valuable discussions.


\begin{thebibliography}{10}

\bibitem{majewski2007non}
W.A. Majewski.
\newblock On non-completely positive quantum dynamical maps on spin chains.
\newblock {\em Journal of Physics A: Mathematical and Theoretical}, 40:11539,
  2007.

\bibitem{shaji2005s}
A.~Shaji and E.C.G. Sudarshan.
\newblock Who's afraid of not completely positive maps?
\newblock {\em Physics Letters A}, 341(1-4):48--54, 2005.

\bibitem{Majewski:2010fk}
W.A. Majewski.
\newblock On the structure of positive maps; finite dimensional case.
\newblock {\em Journal of Mathematical Physics}, 53, 2012.

\bibitem{woronowicz1976positive}
S.L. Woronowicz.
\newblock Positive maps of low dimensional matrix algebras.
\newblock {\em Reports on Mathematical Physics}, 10(2):165--183, 1976.

\bibitem{choi1975completely}
M.D. Choi.
\newblock Completely positive linear maps on complex matrices.
\newblock {\em Linear algebra and its applications}, 10(3):285--290, 1975.

\bibitem{grothendieck1955produits}
A.~Grothendieck.
\newblock Produits tensoriels topologiques et espaces nucl{\'e}aires.
\newblock In {\em Mem. Amer. Math. Soc.}, number~16. American Mathematical
  Society, 1955.

\bibitem{wickstead1973linear}
A.W. Wickstead.
\newblock {\em Linear Operators Between Partially Ordered Banach Spaces and
  Some Related Topics}.
\newblock PhD thesis, Chelsea College, 1973.

\bibitem{jamiokowski1972linear}
A.~Jamio{\l}kowski.
\newblock Linear transformations which preserve trace and positive
  semidefiniteness of operators.
\newblock {\em Reports on Mathematical Physics}, 3(4):275--278, 1972.

\bibitem{stormer1986extension}
E.~St{\o}rmer.
\newblock Extension of positive maps into {B(H)}.
\newblock {\em Journal of Functional Analysis}, 66(2):235--254, 1986.

\bibitem{Gregg:2008lr}
M.C. Gregg.
\newblock {\em {$C^*$}-extreme Points Of The Generalized State Space Of A
  Commutative {$C^*$}-algebra}.
\newblock PhD thesis, University of Nebraska, Lincoln, Nebraska, May 2008.

\bibitem{arveson1969subalgebras}
W.B. Arveson.
\newblock Subalgebras of {$C^*$}-algebras.
\newblock {\em Acta Mathematica}, 123(1):141--224, 1969.

\bibitem{Stormer:1963lr}
E.~St{\o}rmer.
\newblock Positive linear maps of operator algebras.
\newblock {\em Acta Mathematica}, 110(1):233--278, 1963.

\bibitem{farenick1997c}
D.R. Farenick and P.B. Morenz.
\newblock {$C^*$}-extreme points in the generalised state spaces of a
  {$C^*$}-algebra.
\newblock {\em Transactions of the American Mathematical Society},
  349(5):1725--1748, 1997.

\bibitem{cho1992generalized}
S.J. Cho, S.H. Kye, and S.G. Lee.
\newblock Generalized {C}hoi maps in three-dimensional matrix algebra.
\newblock {\em Linear algebra and its applications}, 171:213--224, 1992.

\bibitem{tanahashi1988indecomposable}
K.~Tanahashi and J.~Tomiyama.
\newblock Indecomposable positive maps in matrix algebras.
\newblock {\em Canad. Math. Bull}, 31(3):308--317, 1988.

\bibitem{choi1973positive}
M.D. Choi.
\newblock {\em Positive linear maps on C*-algebras}.
\newblock PhD thesis, University of Toronto., 1973.

\bibitem{bratteli2003operator}
O.~Bratteli and D.W. Robinson.
\newblock {\em Operator Algebras and Quantum Statistical Mechanics}, volume~1.
\newblock Springer, 2003.

\bibitem{cubitt2007dimension}
T.~Cubitt, A.~Montanaro, and A.~Winter.
\newblock On the dimension of subspaces with bounded {S}chmidt rank.
\newblock {\em Journal of Mathematical Physics}, 49:022107, 2008.

\bibitem{parthasarathy2004maximal}
K.R. Parthasarathy.
\newblock On the maximal dimension of a completely entangled subspace for
  finite level quantum systems.
\newblock {\em Proceedings Mathematical Sciences}, 114(4):365--374, 2004.

\bibitem{marciniak2011rank}
M.~Marciniak.
\newblock Rank properties of exposed positive maps.
\newblock {\em Arxiv preprint arXiv:1103.3497, accepted for print Linear and
  Multilinear Algebra}, 2011.

\bibitem{chruscinski2012exposed}
D.~Chru{\'s}ci{\'n}ski and G.~Sarbicki.
\newblock Exposed positive maps: a sufficient condition.
\newblock {\em Journal of Physics A: Mathematical and Theoretical}, 45:115304,
  2012.

\bibitem{MR813392}
A.~Guyan Robertson.
\newblock Positive projections on {$C^\ast$}-algebras and an extremal positive
  map.
\newblock {\em J. London Math. Soc. (2)}, 32(1):133--140, 1985.

\bibitem{sarbicki2012class}
G.~Sarbicki and D.~Chru{\'s}ci{\'n}ski.
\newblock A class of exposed indecomposable positive maps.
\newblock {\em arXiv preprint arXiv:1201.5995}, 2012.

\end{thebibliography}
\end{document}